%% file: main.tex
\documentclass[sigplan,nonacm]{acmart}
\hypersetup{keeppdfinfo}

\AtBeginDocument{%
  }

\usepackage{amsthm}
\usepackage{enumitem}
\usepackage{nicefrac}
\usepackage{algorithm}
\usepackage{xcolor}
\usepackage{multirow}
\usepackage[shortcuts]{extdash}
\usepackage{flushend}

\setcopyright{none}
\renewcommand\footnotetextcopyrightpermission[1]{}
\newtheorem{theorem}{Theorem}
\newtheorem{proposition}{Proposition}

\theoremstyle{definition}
\newtheorem{definition}{Definition}

\begin{document}

\title[COMPASS-ABS]{COMPASS-ABS: Reducing Fragmentation in Shared GPU Clusters for Deep Learning Training Workloads}

\input{authors}

\begin{abstract}
With the rapid advancement of deep learning technology, shared GPU clusters receive an increasing number of deep learning training (DLT) jobs. Yet resource fragmentation make such clusters underutilized and forces the DLT jobs running on them to endure long turnaround times. Extensive research has been devoted to quantifying fragmentation and developing scheduling algorithms that alleviate its impact. However, existing fragmentation measures break down in the absence of workload distribution information, while current schedulers cannot continuously maintain resource fragmentation at a low level. To tackle these problems, we first introduce Scheduler-Induced Fragmentation (SIF), a metric built on the notion of partial-nodes that is independent of historical workload knowledge. We then propose COMPASS-ABS, which employs the COMPact-ASSured (COMPASS) algorithm to confine the cluster state within a tight Anchor-Based Space (ABS), whose construction fully leverages the topological alignment between dominant workload size and
node capacity. Moreover. We also prove that it ensures SIF is bounded by $\frac{2}{N}$ under a workload composition condition that matches both theory and production. Evaluations implemented on a physical cluster and a simulated cluster demonstrate COMPASS-ABS effectiveness at improving resource utilization, reducing DLT job completion time by reducing fragmentation.
\end{abstract}

\ccsdesc[300]{Computer systems organization~Cloud computing}
\ccsdesc[500]{Theory of computation~Scheduling algorithms}

\keywords{Deep learning training, shared GPU cluster, fragmentation, resource scheduling}

\maketitle

\input{section1}

\input{section2}

\input{section3}

\input{section4}

\input{section5}

\input{section6}

\input{section7}

\input{section8}

\bibliographystyle{ACM-Reference-Format}
\bibliography{references}

\appendix

\input{Proof/main_proofs}

\end{document}

%% file: authors.tex
\author{Yukai Zhou}

\author{Hongfan Wu}
\authornote{Corresponding author.}

\renewcommand{\shortauthors}{Yukai Zhou and Hongfan Wu}
\hypersetup{pdfauthor={Yukai Zhou, Hongfan Wu}}

%% file: section1.tex
\section{Introduction}
\label{sec:intro}
\vspace{-2pt}

Deep learning has developed rapidly in recent years, driven by larger model frameworks, more complex training pipelines, and the rise of large language models\citep{narayanan2021efficient,jiang2024megascale}. Training these models requires substantial GPU resources, so industrial labs and cloud providers operate large shared GPU clusters where users submit deep learning training (DLT) jobs\citep{jeon2019analysis,weng2022mlaas,hu2021characterization,nvidia2026superpod}. Yet production clusters exhibit both low GPU cluster utilization and long turnaround time for DLT jobs, which stems from the resource fragmentation rather the capacity shortage\citep{weng2023beware,lao2025cafgd}.

\textbf{Resource fragmentation arises from DLT job placement constraints.}
A DLT job consists of multiple homogeneous instances, and the GPUs assigned to each instance need to lie within a single node to preserve fast intra-instance communication\citep{gu2019tiresias,choudhury2024mast}. This intra-node constraint can prevent admission even when the cluster has enough available GPUs in total. Consider a cluster with three nodes that currently have 3, 5, and 2 free GPUs, i.e. 10 free GPUs in total. A job with two 4-GPU instances demanding 8 GPUs in total cannot be admitted because only one node can host a 4-GPU instance. The fragmented distribution of idle GPUs blocks a placement that would appear feasible under an aggregate-capacity view. So resource fragmentation leads to low GPU utilization despite pending jobs awaiting processing, causing prolonged job completion time due to long queueing latency.

This observation raises two central questions: (1) how to quantify the resource fragmentation? (2) how to solve the problem of GPU cluster underutilization and long completion time of DLT jobs caused by resource fragmentation? Existing fragmentation metrics are statistical measures of fragmentation. FGD\citep{weng2023beware} defines a workload-distribution-aware fragmentation score, while CAFGD\citep{lao2025cafgd} refines this idea by weighting long-term and short-term distributions. These metrics rely on workload history knowledge and fail to measure fragmentation in its absence.

Existing schedulers for scheduling DLT jobs in shared GPU clusters can be grouped into two categories. \textbf{Non-migration methods place each job at admission and never revisit the decision}: ElasticFlow\citep{gu2023elasticflow} formulates placement as a bin-packing problem and applies best-fit-style heuristics, while FGD\citep{weng2023beware} and CAFGD\citep{lao2025cafgd} greedily optimize their fragmentation metrics at first admission time. These methods react only to the current state and provide no mechanism to repair fragmentation that has existed. \textbf{\allowbreak Migration-based approaches reorganize running DLT jobs through checkpoint-based migration}. Gandiva\citep{xiao2018gandiva} only suggests migrating running jobs when cluster status changes, without specifying the implementation details like when and how to migrate. FFT\citep{mo2025fft} and DRR\citep{wu2025defrag} adopt a round-driven strategy whose defragmentation
relies on global migration at the beginning of each round, which cannot mitigate the resource fragmentation accumulated within each round. More fundamentally, they all cannot guarantee low fragmentation throughout the online scheduling process.

Motivated by these limitations, we introduce \textbf{Scheduler-Induced Fragmentation (SIF)} to measure dynamic resource fragmentation and design the \textbf{COMPASS-ABS scheduler} for DLT jobs in shared GPU clusters. Our contributions are summarized as follows:

\begin{itemize}
    \item \textbf{Scheduler-Induced Fragmentation (SIF).} Building on the notion of partial-nodes, we decompose observed fragmentation into two parts: an inherent component forced by the intra-node placement constraint, and an additional component caused by the scheduler's placement decisions. SIF measures only the latter, which no longer requires the workload history distribution information.
    \item \textbf{COMPASS-ABS scheduler.} We firstly define \textbf{ABS} (Anchor-Based Space), a constrained feasible cluster state domain whose structure is tailored jointly to the dominance of power-of-two GPU demands and to their topology alignment with node capacity. Multi-GPU instances are placed as anchors in slots of width 2, 4, or 8 GPUs, while 1-GPU instances are managed through a pool and used as fillers for saturating slot vacancies. The \textbf{COMPASS}(COMPact-ASSured) algorithm maintains cluster change within this space in a compact manner through three operators: $\mathsf{Place}$, $\mathsf{Remove}$, and $\mathsf{Compact}$. Under the Workload Composition Condition (WCC) stated in Theorem~\ref{thm:compactness}, COMPASS-ABS guarantees $\mathrm{SIF}^{\pi}(t) \le \frac{2}{N}$ at all times.
    \item \textbf{Evaluation.} We conduct comprehensive experiments to validate the scheduling performance of COMPASS-ABS, including large-scale simulations on two production traces and a physical-cluster deployment with a synthetic DLT workload. COMPASS-ABS attains the highest GPU utilization and the shortest job completion time among all state-of-the-art baselines by maintaining the lowest fragmentation level, and sustains its strength under varying cluster load, workload composition, and migration cost, even when the Workload Composition Condition is violated.
\end{itemize}

%% file: section2.tex
\section{Background}
\label{sec:motivation}
\label{sec:background}

This section reviews the task structure and communication patterns of DLT jobs, together with the architectural properties of physical GPU clusters. We point out two system characteristics that motivate the scheduler design in Section~\ref{sec:method}: the per-instance GPU demand $g_i$ is concentrated on power-of-two values, and the basic communication domain in production clusters is a node of 8 GPUs.

\subsection{Deep Learning Training Jobs}
\label{sec:bg-dlt-jobs}

As deep learning models continue to extend, the memory capacity and compute throughput of a single GPU are insufficient for an increasing fraction of training workloads. Modern frameworks therefore support training across multiple GPUs\citep{narayanan2021efficient,huang2019gpipe,li2020pytorchdistributed}. The growth of training datasets drives the use of data parallelism, where the training program is replicated across multiple instances processing different mini-batches. Thus, a DLT job is composed of multiple homogeneous instances that run the same model and carry the same per-instance GPU demand $g_i$\citep{choudhury2024mast}.

Within each instance, multiple GPUs may collaborate through tensor parallelism and other mechanisms to execute the forward and backward passes\citep{narayanan2021efficient,lepikhin2020gshard,narayanan2019pipedream}. These mechanisms frequently invoke collective communication primitives such as All-Reduce, making the instance sensitive to communication latency\citep{hwang2023tutel}. Tensor parallelism partitions large parameter matrices across GPUs, and common tensor-parallel degrees are 2, 4, and 8. Collective implementations such as ring, recursive-doubling, and tree-based All-Reduce are structurally defined on power-of-two participant counts, while non-power-of-two configurations lead to padding or special-case handling\citep{thakur2005optimization}. These properties make $g_i \in \{1,2,4,8\}$ a natural operating regime, which is consistent with production trace characteristics\citep{jeon2019analysis,weng2022mlaas,hu2021characterization}.

Beyond intra-instance communication, different instances of the same job execute data-parallel synchronization at the end of each iteration to apply a synchronous parameter update, which requires all instances to run concurrently. Thus, schedulers for DLT workloads commonly treat each job as an all-or-nothing allocation unit, which imposes a gang scheduling requirement\citep{gu2019tiresias,ye2024survey}.

\subsection{Physical Cluster}
\label{sec:bg-physical-cluster}

Modern GPU clusters comprise multiple nodes, each of which serves as a primary high-bandwidth communication domain in the network topology\citep{jiang2024megascale,nvidia2026superpod}. GPUs on the same node communicate through NVLink at terabyte-per-second bandwidth, whereas inter-node communication involves lower bandwidth and switch-level congestion. For the latency-sensitive intra-instance communication described above, all GPUs assigned to an instance must belong to the same node, which makes the available GPUs non-interchangeable\citep{xiao2018gandiva,weng2023beware}. This node-locality constraint is the structural source of the fragmentation problem studied in this paper.

Each node contains eight GPUs, matching the configuration reported in production clusters\citep{nvidia2017dgx1,jeon2019analysis,weng2022mlaas,hu2021characterization}. The eight-GPU design has a hardware basis: DGX-1 connects eight GPUs as a hybrid cube-mesh, and later DGX/HGX systems preserve eight GPUs as a server-level building block while improving the internal fabric with NVSwitch\citep{nvidia2017dgx1,li2020evaluating,nvidia2026superpod}. This eight-GPU node is an industry-wide convention rather than an NVIDIA-specific choice, as the vendor-neutral OCP OAM Universal Baseboard standard\citep{ocp_ubb} hosts eight accelerators per baseboard and mainstream accelerator servers adopt the same $G = 8$ layout, including AMD Instinct\citep{amd_mi300x}, Intel Gaudi\citep{intel_gaudi3}, Huawei Ascend\citep{huawei_atlas800}, and Meta's Grand Teton\citep{meta_grandteton}.

%% file: section3.tex
\section{Scheduler-Induced Fragmentation}
\label{sec:sif}

This section gives the measure of resource fragmentation without relying on history trace information. In \S\ref{sec:avoidable-forced}, we introduce the notion of partial-nodes, which are the principal carrier of fragmented resource, and separate them into two mutually exclusive proportions. Building on this distinction, \S\ref{sec:decomposing} then defines a new measure called Scheduler-Induced Fragmentation (SIF), computed solely from the resource demand of currently running jobs in the cluster.

\subsection{Key Insight: Avoidable vs.\ Forced Partial Nodes}
\label{sec:avoidable-forced}

We firstly define a node that is non-empty yet assigned fewer GPUs than its full capacity as a \textbf{partial node}. A fragmented placement disperses free GPUs across many partial nodes and causes contiguous GPU resources to become scarce, whereas a well-designed placement will reduce the partial-node count to release the fragmented GPUs trapped inside partial nodes and return them to larger contiguous GPU blocks, which is friendly to all DLT jobs.

This concern becomes most pressing when the cluster receives large jobs. Consider a job consisting of $w$ 8-GPU instances which require $w$ fully-empty nodes on the cluster before it can be admitted. Consider a cluster of $N$ nodes in which the current fragmented placement consumes $N_1$ nodes, whereas a compact placement of the same set of running DLT jobs would consume only $N_2 < N_1$. The compact placement therefore leaves $N - N_2$ nodes entirely empty, while the fragmented one leaves only $N - N_1$. Whenever the requested job size satisfies $N - N_1 < w \le N - N_2$, that job is admissible under the compact placement but is blocked under the fragmented one, even though the cluster's aggregate free-GPU count exceeds $8 w$ in both situations. This large DLT job blocking results from fully-empty-node availability: $N_1 - N_2$ nodes trapped in partial occupancy could otherwise serve as the continuous GPU blocks that large jobs demand.

We therefore set our primary objective for mitigating fragmentation as \textbf{minimizing the number of partial nodes}, which can equivalently be stated as minimizing the total number of nodes in use. We adopt the former formulation throughout the remainder of this paper. However, not every partial node is produced by the scheduler: some cannot be eliminated by any scheduler, while others arise from a suboptimal scheduling decision and can be eliminated by right placement and migration. Then we demonstrate the distinction through two illustrative examples.

\textbf{Forced partial nodes.} Suppose the cluster currently has only one partial node, which hosts a single $4$-GPU instance, a DLT job consisting of two $5$-GPU instances waits at the head of queue. Due to intra-node constraint, neither of the two new instances can utilize the 4 free GPUs on this partial node. The only feasible placement is opening two other empty nodes, each hosting a single $5$-GPU instance with $3$ GPUs unused, so that the cluster now contains three partial nodes in total. No scheduler could have done better in this situation, because these three partial nodes are byproducts of systemic constraints rather than incorrect scheduling.
\begin{sloppypar}
\textbf{Avoidable partial nodes.} Consider two nodes that are each fully packed without free GPUs, where a single job has placed one $4$-GPU instance on each node and the remaining $4$ GPUs of each node are occupied by unrelated instances. When this job completes, both nodes simultaneously release a $4$-GPU block and the cluster gains two partial nodes in a single event. This fragmentation is not structurally necessary,  because the scheduler could migrate all instances on the second node into the idle region of the first, after which the cluster would contain zero partial nodes. The two partial nodes here are avoidable, because their existence is entirely attributable to the scheduler's decision.
\end{sloppypar}

\subsection{Decomposing Fragmentation: Inherent and Scheduler-Induced}
\label{sec:decomposing}

Motivated by \S\ref{sec:avoidable-forced}, we now formally split the fragmentation into two complementary components, the portion that any scheduler must leave behind, and the portion introduced by chosen scheduler $\pi$.

\textbf{Inherent Fragmentation.} We denote the set of DLT jobs running on the cluster at time $t$ by $\mathcal{J}(t)$. Each job $i \in \mathcal{J}(t)$ consists of $w_i$ homogeneous instances, 
each demanding $g_i$ GPUs, where $g_i \in \{1, \dots, G\}$ and $G = 8$. The instantaneous demand structure of running DLT jobs is then formalized as $\mathcal{I}(t) = \{(g_i, w_i) : i \in \mathcal{J}(t)\}$. In order to isolate the fragmentation originating from constraints alone, we construct a static bin-packing problem in which each DLT job is treated as $w_i$ items of size $g_i$ and each node as a bin of capacity $G$. The question then becomes: what is the minimum partial-node count achievable over all feasible packings of these items into bins? Formally, for any feasible packing $P$ that satisfies the intra-node constraint, we let $\nu(P)$ denote the number of partial nodes in $P$ and write $\Phi^\star(t) := \min_{P \text{ feasible}} \nu(P)$ for the minimum partial-node count achievable across all such packings. The Inherent fragmentation is then the normalised share
\begin{equation}
F^\star(t) \;:=\; \frac{\Phi^\star(t)}{N} \;=\; \frac{\min_{P \text{ feasible}} \nu(P)}{N}.
\label{eq:fstar}
\end{equation}
Because $F^\star(t)$ depends only on the real time resource demand $\mathcal{I}(t)$ and is independent of arrival order, placement history, or migration operation, it serves as an \textbf{unavoidable lower bound} on the partial-node share. A heuristic algorithm for computing $\Phi^{\star}(t)$ (and therefore $F^\star(t)$) is provided in Appendix~\ref{app:sec:algorithm-SIFOpt}.

\textbf{Scheduler-Induced Fragmentation.} Under an online policy $\pi$, the placement of all jobs in the cluster observed at time $t$ is itself a feasible packing $P^\pi(t)$ of $\mathcal{I}(t)$, and we let $F^\pi(t) := \nu(P^\pi(t))/N$ denote its partial-node share, which represents the gross fragmentation. We define the Scheduler-Induced Fragmentation (SIF) by subtracting the inherent fragmentation from the gross fragmentation:
\begin{equation}
\mathrm{SIF}^\pi(t) \;:=\; F^\pi(t) - F^\star(t) \;\ge\; 0,
\label{eq:sif}
\end{equation}
which represents the percentage of avoidable partial nodes that the policy $\pi$ has produced. This definition makes SIF a more reasonable fragmentation measure than the statistical metrics reported in FGD and CAFGD. It builds upon the resource demand of the running jobs rather than further historical information about the workload resource demand distribution. Moreover, it accurately quantifies the proportion of fragmentation that is introduced by the scheduler. Therefore, it motivates us to design a scheduler that persistently keeps $\mathrm{SIF}^\pi(t)$ at a low level.

%% file: section4.tex
\section{COMPASS-ABS Scheduler}
\label{sec:method}

This section presents our COMPASS-ABS scheduler that employs the COMPASS algorithm to maintain the cluster evolving in the compact ABS domain. In Section~\ref{sec:design-rationale}, we introduce our design motivation for the innovative scheduler; In Section~\ref{sec:state-space}, we define a feasible space ABS for the cluster to operate in, driven by the characteristic of workload composition and cluster topology; In Section~\ref{sec:scheduling-policy}, we present the COMPASS algorithm that constrains the cluster state within the ABS in a consolidated manner. In Section~\ref{sec:guarantees}, we further provide theoretical guarantee on the ability of our scheduler to keep  $\mathrm{SIF}$ low at minimal computational cost.

\subsection{Design Rationale}
\label{sec:design-rationale}

\textbf{Observation 1: Power-of-two-sized multi-GPU instances scheduled by the best-fit policy when node capacity is 8 keep $\mathrm{SIF}\le \frac{2}{N}$.}
Before examining fragmentation in detail, we set aside the instances whose $g_i=1$, denoted as $\mathcal{I}^{=1}$, because they never create fragmentation but rather act as the natural filler units to mitigate fragmentation by utilizing the scattered idle GPUs left on partial nodes. Restricting attention to the multi-GPU instances, denoted as $\mathcal{I}^{\ge 2}$, we exploit a divisibility alignment between the instance demand and the node capacity: every $g_i \in \{2, 4, 8\}$ divides $G = 8$ exactly, so any combination of such instances either fully fills a node or leaves a residue GPU block of size in $\{2, 4, 6\}$. Under this alignment, a best-fit policy that places each instance on the node leaving the smallest remaining capacity after placement, together with a lightweight Compact pass that consolidates partial nodes after every departure, retains $\mathrm{SIF}^{\pi}(t)\le \frac{2}{N}$ at all times. The formal proof, including the matching tightness construction, is deferred to Appendix~\ref{app:sec:bf-power-of-two}. We accordingly name the multi-GPU instances with $g_i \in \{2,4,8\}$ as \textbf{Fragmentation-Friendly instances} and denote them by $\mathcal{I}^{\mathrm{FF}}$, whereas instances whose demand falls in $\{3,5,6,7\}$ are called \textbf{Fragmentation-Unfriendly instances}, denoted by $\mathcal{I}^{\mathrm{FUF}}$, so all multi-GPU instances can be partitioned into  $\mathcal{I}^{\ge 2} \;:=\; \mathcal{I}^{\mathrm{FF}} \;\cup\; \mathcal{I}^{\mathrm{FUF}}$.

\textbf{Observation 2: Fragmentation-unfriendly instances account for only a small share of real DLT workloads.}
The structural reasons in \S\ref{sec:bg-dlt-jobs} make $g_i \in \{1,2,4,8\}$ the natural operating regime for DLT instances, so non-power-of-two instance sizes $g_i \in \{3,5,6,7\}$ account for only a small fraction of the instance population in production traces. Table~\ref{tab:instance-mix} confirms this claim by reporting the average per-type instance shares observed on the GFS\citep{duan2026gfs} and Venus\citep{hu2021characterization} traces, where the $g_i=1$ and $g_i=8$ classes together dominate the trace and the $\mathcal{I}^{\mathrm{FUF}}$ retains only a marginal share.

These observations lead to two design ideas.

\textbf{Design idea 1: Round every fragmentation-unfriendly instance into a fragmentation-friendly composite block with fillers.}
Observation 1 already points out that a best-fit placement keeps $\mathrm{SIF} \le \frac{2}{N}$ once the workload is restricted to power-of-two-sized instances, and we therefore want to extend the same guarantee to the fragmentation-unfriendly population. To this end, we round instance in $\mathcal{I}^{\mathrm{FUF}}$ into a friendly-sized composite block whose total size falls in $\{2, 4, 8\}$ by combining it with smaller instances. Then the scheduler only sees fragmentation-friendly-shaped blocks rather than a heterogeneous mixture of friendly and unfriendly demands under this transformation, so that the placement decisions on the entire multi-GPU population $\mathcal{I}^{\ge 2}$ reduce to the $\mathcal{I}^{\mathrm{FF}}$-only case.

\textbf{Design idea 2: Employ 1-GPU instances to assemble the fragmentation-friendly composite blocks and manage them in a dedicated pool.}
As 1-GPU instances largely outweigh those in $\mathcal{I}^{\mathrm{FUF}}$ and they are the natural filler units to mitigate fragmentation, we accordingly use 1-GPU instances to fill the residual GPUs inside each composite block, namely the gap between the original unfriendly instance and its rounded friendly block. Due to the dynamic arrival and departure of $\mathcal{I}^{\mathrm{FUF}}$ instances, the requests for fillers change over time, and the cluster needs to adaptively migrate 1-GPU instances to the right place to maintain the integrity of the composite block. This motivates us to place all 1-GPU instances in a unified pool for better management.

\begin{table}[h]
  \caption{Per-class instance share across two DLT traces.}
  \label{tab:instance-mix}
  \small
  \setlength{\tabcolsep}{6pt}
  \begin{tabular}{lcccc}
    \toprule
    Trace  & $\mathcal{I}^{=1}$ & $\mathcal{I}^{\mathrm{FF}}$ & $\mathcal{I}^{\mathrm{FUF}}$ & $\mathcal{I}^{=8}$ \\
    \midrule
    GFS    & 64.16\% &  6.07\% & 0.05\% & 29.73\% \\
    Venus  & 37.68\% &  9.44\% & 0.44\% & 52.44\% \\
    \bottomrule
  \end{tabular}
\end{table}

\subsection{Anchor-Based Space (ABS)}
\label{sec:state-space}

Driven by the design ideas above, we define a structured state space for the cluster that we call the \textbf{Anchor-Based Space (ABS)}, within which the placement of all DLT jobs is required to remain throughout the cluster's evolution. Then we formalize the key objects in the ABS domain.

\begin{definition}[Slot]
\label{def:slot}
A \textbf{slot}, denoted by $s$, is a GPU block inside a node. Each slot has a fixed width $w(s) \in \{2,4,8\}$. When a multi-GPU instance of demand $g$ is assigned to a slot, the slot width is chosen according to
\begin{equation}
w^\star(g) \;=\;
\begin{cases}
2, & g = 2, \\
4, & g \in \{3,4\}, \\
8, & g \in \{5,6,7,8\}.
\end{cases}
\label{eq:wstar}
\end{equation}
The instances belonging to $\mathcal{I}^{\mathrm{FF}}$ will directly saturate width-matched slots, while counterpart instances will be assigned to rounded-up slots.
\end{definition}

\begin{definition}[Anchor and Slot Vacancy]
\label{def:anchor}
On each slot $s$, there is exactly one instance $a(s) \in \mathcal{I}^{\ge 2}$ called the \textbf{anchor} of $s$ that occupies $g_{a(s)}$ GPUs in that slot, where $g_{a(s)}$ is the demand of that instance. The remaining $w(s)-g_{a(s)}$ GPUs in the slot are called the \textbf{slot vacancy} of $s$. Whenever $a(s) \neq \varnothing$, the slot width satisfies $w(s)=w^\star(g_{a(s)})$.
\end{definition}

\begin{definition}[Filler and Capacity Invariant]
\label{def:filler}
When the slot vacancy of $s \neq 0$, we only employ 1-GPU instances to saturate the slot vacancy, which are called \textbf{fillers}. Representing the set of fillers currently attached to $s$ as $F(s)$, we require
\begin{equation}
g_{a(s)} + |F(s)| \;\le\; w(s).
\label{eq:capacity-invariant}
\end{equation}
Equality means the slot is fully used.
\end{definition}

\begin{definition}[Slotted node and slot configuration set]

\label{def:slotted-node}

Node $n$ that hosts at least one slot is called a \textbf{slotted node}, and we denote the multiset of slots currently profiled on $n$ by $S(n)$, where the slots in $S(n)$ are non-overlapping. The residual $G - \sum_{s \in S(n)} w(s)$ GPUs on $n$ have not yet been profiled as any slot, and a new slot is created over them only when an anchor is placed on $n$. The collection of all feasible slot profiling plans forms the \textbf{slot configuration set}

\begin{equation}
\mathsf{SlotConf} \;:=\; \Bigl\{\, \{s_1, \ldots, s_k\} \;\Big|\;
s_i \in \{2, 4, 8\}\ \forall\, i,\;
\sum_{i=1}^{k} s_i \le G \,\Bigr\},
\label{eq:slotconf}
\end{equation}
where $S(n) \in \mathsf{SlotConf}, \forall n$.
\end{definition}

\begin{definition}[Pool node]
\label{def:pool-node}
A \textbf{pool node} is a node dedicated to 1-GPU instances. For a pool node $n$, we write $P(n) \subseteq \mathcal{I}^{=1}$ for the set of 1-GPU instances currently on $n$, with $|P(n)| \le G$. A 1-GPU instance on a pool node is not acting as a filler but can be moved to the slot vacancy when necessary.
\end{definition}

Building on the definitions above, every node in the cluster falls into three mutually exclusive types: an empty node holds no instances, a slotted node whose $S(n) \in \mathsf{SlotConf}$, and a pool node whose $P(n) \subseteq \mathcal{I}^{=1}$. We accordingly collect these three labels into the node-type set
\begin{equation}
\mathsf{T} \;:=\; \{\,\mathsf{empty},\; \mathsf{pool},\; \mathsf{slotted}\,\},
\label{eq:node-type}
\end{equation}
and write $t(n) \in \mathsf{T}$ for the type currently assigned to node $n$. This coarse classification underlies the formal definition of the feasible cluster state space $\Sigma$ given below.

\textbf{Anchor-Based Space (ABS) $\Sigma$.}\ 
For each node $n$, the internal state $\sigma_n$ is determined by $t(n)$:
\begin{equation}
\sigma_n \;=\;
\begin{cases}
\bot, & t(n)=\mathsf{empty}, \\[2pt]
P(n) \subseteq \mathcal{I}^{=1},\; |P(n)| \le G, & t(n)=\mathsf{pool}, \\[2pt]
\bigl((a_j,F_j)\bigr)_{j=1}^{k}, & t(n)=\mathsf{slotted}.
\end{cases}
\label{eq:sigma-n}
\end{equation}
where $a_j \in \mathcal{I}^{\ge 2}$, $F_j \subseteq \mathcal{I}^{=1}$ and together satisfy the capacity invariant in \eqref{eq:capacity-invariant}. The cluster-level state space is therefore
\begin{equation}
\Sigma \;=\; \Bigl\{\, \sigma = \bigl(\pi,(\sigma_n)_{n \in \mathcal{N}}\bigr) \;\Big|\; \pi: \mathcal{N} \to \mathrm{T},\; \sigma_n \text{ as above} \,\Bigr\}.
\label{eq:Sigma}
\end{equation}

This state-space design is organized entirely around the notion of anchor, which provides a unified scheduling unit for both fragmentation-friendly and fragmentation-unfriendly multi-GPU instances. Starting from this concept, the gap between a slot's width and its anchor's actual demand naturally introduces fillers to absorb the residual capacity inside each slot, and the practical need to migrate these fillers flexibly across the cluster in turn motivates the design of pool nodes as the dedicated pool for 1-GPU instances. Therefore, we name this state space $\Sigma$ as \textbf{Anchor-Based Space} (ABS).

\subsection{COMPASS Algorithm}
\label{sec:scheduling-policy}

We propose the COMPact-ASSured (COMPASS) algorithm that consists of three operators, the $\mathsf{Place}$ and $\mathsf{Remove}$ operators together guarantee that the cluster state lies within ABS $\Sigma$ as the job is admitted and completed, while the $\mathsf{Compact}$ ensures the compactness of ABS $\Sigma$.

Each DLT job is admitted only when all of its instances can be placed simultaneously and is released as a whole upon completion. We therefore formulate scheduling decisions at the instance granularity, where every job-level scheduling is translated into a sequence of per-instance operations whose union realises the corresponding job-level decision, and for an instance $i$ we write $g_i$ for its GPU demand.

\subsubsection{Common Notation}

We write $\sigma \in \Sigma$ for the current cluster state, $\mathcal{D}$ for the set of legal destinations where a single instance can be assigned, and $\mathcal{M}^\star$ for the space of finite migration lists. A destination $d \in \mathcal{D}$ takes one of four forms: a slotted node where a new slot of the required width can be profiled, a slot vacancy of an anchored slot, a partial pool node, or an empty node waiting to be opened. A single migration is written as a pair $(i, d)$ that moves instance $i$ to destination $d$, and the update notation $\sigma \oplus M$ means that the migrations in $M \in \mathcal{M}^\star$ are applied to $\sigma$ in sequence.

\subsubsection{Place Operator}

The placement operator carries the signature $\mathsf{Place}: \Sigma \times \mathcal{I} \to \mathcal{D} \times \mathcal{M}^\star$, which takes the current cluster state $\sigma$ together with an arriving instance $i$ and returns the destination $d$ where $i$ is to be installed alongside the migration list $M$.

If $i \in \mathcal{I}^{\ge 2}$, the scheduler searches for a slotted node with at least $w^\star(g_i) \in \{2, 4, 8\}$ free GPUs, or opens a fresh empty node if none exists; it then profiles a width-$w^\star(g_i)$ slot and places $i$ as the anchor. When $i \in \mathcal{I}^{\mathrm{FUF}}$, the scheduler migrates up to $w^\star(g_i) - g_i$ 1-GPU instances from the pool nodes as fillers, subject to filler availability in the pool.

If $i \in \mathcal{I}^{=1}$, the scheduler places $i$ following a strict priority order: into an existing slot vacancy if one exists; otherwise into a partial pool node; and only resorts to a newly opened pool node when neither exists. This order lets the 1-GPU instance firstly mitigate existing fragmentation before activating any new slotted node.

\begin{center}
\fbox{\begin{minipage}{\dimexpr\linewidth-2\fboxsep-2\fboxrule\relax}
\small
\noindent\textbf{Pseudocode of $\mathsf{Place}(\sigma, i)$.}\\[2pt]
\textbf{Input:}\ state $\sigma$; arriving instance $i$ with demand $g_i$ \\
\textbf{Output:}\ destination $d \in \mathcal{D}$; migration list $M \in \mathcal{M}^\star$

\medskip
\noindent 1:\ $M \leftarrow \emptyset$ \\
\noindent 2:\ \textbf{if}\ $g_i \ge 2$\ \textbf{then} \\
\noindent 3:\ \ \ $w \leftarrow w^\star(g_i)$ \\
\noindent 4:\ \ \ $\mathcal{N}_{\mathrm{cand}} \leftarrow \{\, n : t(n)=\mathsf{slotted},\; G - \sum_{s \in S(n)} w(s) \ge w \,\}$ \\
\noindent 5:\ \ \ \textbf{if}\ $\mathcal{N}_{\mathrm{cand}} \neq \emptyset$\ \textbf{then}\ $n^\star \leftarrow \arg\min_{n \in \mathcal{N}_{\mathrm{cand}}}\bigl(G - \sum_{s \in S(n)} w(s)\bigr)$ \\
\noindent 6:\ \ \ \textbf{else}\ pick any empty node $n^\star$ and set $t(n^\star) \leftarrow \mathsf{slotted}$ \\
\noindent 7:\ \ \ profile a fresh slot $s^\star$ of width $w$ on $n^\star$;\ set $a(s^\star) \leftarrow i$;\ $d \leftarrow s^\star$ \\
\noindent 8:\ \ \ $r \leftarrow w - g_i$ \\
\noindent 9:\ \ \ \textbf{if}\ $r > 0$\ \textbf{then}\ draw up to $r$ 1-GPU instances $\{f_1, \ldots, f_k\}$ from pool nodes;\ $M \leftarrow \{(f_\ell, s^\star) : \ell = 1, \ldots, k\}$ \\
\noindent 10:\ \textbf{else}\ \ \ \emph{(}$g_i = 1$\emph{)} \\
\noindent 11:\ \ \ \textbf{if}\ $\exists\, s$ with $g_{a(s)} + |F(s)| < w(s)$\ \textbf{then}\ $d \leftarrow s$ \\
\noindent 12:\ \ \ \textbf{else if}\ $\exists\, n$ with $t(n)=\mathsf{pool}$ and $|P(n)| < G$\ \textbf{then}\ $d \leftarrow n$ \\
\noindent 13:\ \ \ \textbf{else}\ pick any empty node $n_0$;\ set $t(n_0) \leftarrow \mathsf{pool}$;\ $d \leftarrow n_0$ \\
\noindent 14:\ \textbf{return}\ $(d, M)$
\end{minipage}}
\end{center}

\subsubsection{Remove Operator}

The departure operator carries the signature $\mathsf{Remove}: \Sigma \times \mathcal{I} \to \mathcal{M}^\star$, which returns the migrations triggered when instance $i$ departs from state $\sigma$. 

If $i \in \mathcal{I}^{\ge 2}$ , it serves as the anchor $a(s)$ of slot $s$. The operator dissolves $s$ by releasing $w(s)$ GPUs, and migrates each filler in $F(s)$ into the pool node through the same 1-GPU rule that $\mathsf{Place}$ applies to it; the migration $\mathcal{M} = \emptyset$ when $a(s) \in \mathcal{I}^{\mathrm{FF}}$, since $F(s) = \emptyset$ in that case.

If $i \in \mathcal{I}^{=1}$, the operator releases $i$ from its current location, which is either on a slotted node or a pool node. In the former case, it additionally migrates another 1-GPU instance from the pool to keep the slot saturated.

\begin{center}
\fbox{\begin{minipage}{\dimexpr\linewidth-2\fboxsep-2\fboxrule\relax}
\small
\noindent\textbf{Pseudocode of $\mathsf{Remove}(\sigma, i)$.}\\[2pt]
\textbf{Input:}\ state $\sigma$; departing instance $i$ \\
\textbf{Output:}\ migration list $M \in \mathcal{M}^\star$

\medskip
\noindent 1:\ $M \leftarrow \emptyset$ \\
\noindent 2:\ \textbf{if}\ $g_i \ge 2$\ \textbf{then}\ \ \emph{($i$ is an anchor)} \\
\noindent 3:\ \ \ $s \leftarrow$ the slot whose anchor is $i$ \\
\noindent 4:\ \ \ \textbf{for each}\ $f \in F(s)$\ \textbf{do} \\
\noindent 5:\ \ \ \ \ choose destination $d_f$ for $f$ by the 1-GPU rule of $\mathsf{Place}$ \\
\noindent 6:\ \ \ \ \ append $(f, d_f)$ to $M$ \\
\noindent 7:\ \ \ \textbf{end for} \\
\noindent 8:\ \ \ dissolve $s$ and release its $w(s)$ GPUs on the host node \\
\noindent 9:\ \textbf{else}\ \ \emph{($g_i = 1$)} \\
\noindent 10:\ \ \ detach $i$ from its current slot vacancy or pool node \\
\noindent 11:\ \textbf{return}\ $M$
\end{minipage}}
\end{center}

\subsubsection{Compact Operator}

We denote the compaction operator by $\mathsf{Compact}: \Sigma \to \mathcal{M}^\star$, which runs after each job event, namely the arrival or departure of each DLT job. It performs consolidation via three sequential calls to a per-class subroutine $\mathsf{BestFitDrain}$

Within each invocation, $\mathsf{BestFitDrain}$ identifies the partial hosts currently holding class-$\mathcal{C}$ items, orders them by descending room so that the sparsest source is processed first, and migrates the items of that source one by one into the partial host whose remaining class-$\mathcal{C}$ room is the smallest among those still able to accept the migrated item. For the three concrete classes used by $\mathsf{Compact}$, the pool pass takes 1-GPU instances as items and pool nodes as hosts with room $G - |P(n)|$, while the width-$4$ and width-$2$ passes take anchored slots of the corresponding width as items and slotted nodes as hosts with room $\lfloor D(n)/4 \rfloor$ and $\lfloor D(n)/2 \rfloor$ respectively, where $D(n) = G - \sum_{s \in S(n)} w(s)$.

\begin{center}
\fbox{\begin{minipage}{\dimexpr\linewidth-2\fboxsep-2\fboxrule\relax}
\small
\noindent\textbf{Pseudocode of $\mathsf{Compact}(\sigma)$.}\\[2pt]
\textbf{Input:}\ state $\sigma$ \\
\textbf{Output:}\ migration list $M \in \mathcal{M}^\star$

\medskip
\noindent 1:\ $M \leftarrow \emptyset$ \\
\noindent 2:\ \textbf{for}\ each class $\mathcal{C} \in \langle \mathrm{pool},\, \mathrm{w4},\, \mathrm{w2} \rangle$\ \textbf{do} \\
\noindent 3:\ \ \ $\mathcal{N} \leftarrow \{\, n : 0 < \mathrm{count}_\mathcal{C}(n) < \mathrm{cap}_\mathcal{C}(n) \,\}$ \\
\noindent 4:\ \ \ \textbf{sort}\ $\mathcal{N}$ by $\mathrm{room}_\mathcal{C}(\cdot)$ in descending order \\
\noindent 5:\ \ \ \textbf{while}\ $|\mathcal{N}| \ge 2$\ \textbf{do} \\
\noindent 6:\ \ \ \ \ $n_{\mathrm{src}} \leftarrow \mathrm{head}(\mathcal{N})$;\ $\mathrm{progress} \leftarrow \mathrm{false}$ \\
\noindent 7:\ \ \ \ \ \textbf{for each}\ class-$\mathcal{C}$ item $x$ on $n_{\mathrm{src}}$\ \textbf{do} \\
\noindent 8:\ \ \ \ \ \ \ $\mathcal{S} \leftarrow \{\, n \in \mathcal{N} \setminus \{n_{\mathrm{src}}\} : \mathrm{room}_\mathcal{C}(n) \ge \mathrm{size}(x) \,\}$ \\
\noindent 9:\ \ \ \ \ \ \ \textbf{if}\ $\mathcal{S} = \emptyset$\ \textbf{then break} \\
\noindent 10:\ \ \ \ \ \ $n_{\mathrm{dst}} \leftarrow \arg\min_{n \in \mathcal{S}} \mathrm{room}_\mathcal{C}(n)$ \\
\noindent 11:\ \ \ \ \ \ migrate $x$ from $n_{\mathrm{src}}$ to $n_{\mathrm{dst}}$;\ append to $M$;\ $\mathrm{progress} \leftarrow \mathrm{true}$ \\
\noindent 12:\ \ \ \ \ \textbf{end for} \\
\noindent 13:\ \ \ \ \ \textbf{if}\ $n_{\mathrm{src}}$ has no class-$\mathcal{C}$ item left\ \textbf{then}\ remove $n_{\mathrm{src}}$ from $\mathcal{N}$ \\
\noindent 14:\ \ \ \ \ \textbf{if not}\ $\mathrm{progress}$\ \textbf{then break} \\
\noindent 15:\ \ \ \textbf{end while} \\
\noindent 16:\ \textbf{end for} \\
\noindent 17:\ \textbf{return}\ $M$

\medskip
\noindent\textit{Notation.}\ For class $\mathcal{C} \in \{\mathrm{pool},\, \mathrm{w4},\, \mathrm{w2}\}$, $\mathrm{count}_\mathcal{C}(n)$ is the number of class-$\mathcal{C}$ items currently on $n$, $\mathrm{cap}_\mathcal{C}(n)$ is the maximum number of class-$\mathcal{C}$ items that $n$ can hold, $\mathrm{room}_\mathcal{C}(n) = \mathrm{cap}_\mathcal{C}(n) - \mathrm{count}_\mathcal{C}(n)$, and $\mathrm{size}(x) = 1$ in all three classes. The local variables $n_{\mathrm{src}}$ and $n_{\mathrm{dst}}$ denote, respectively, the source node from which an item is migrated and the destination node into which it is migrated within the current best-fit drain pass.
\end{minipage}}
\end{center}

\begin{algorithm}[ht]
\framebox[\linewidth]{%
\begin{minipage}{0.93\linewidth}
\small
\noindent\textbf{Input:}\ initial state $\sigma_0 \in \Sigma$;\ online event stream $\langle e_1, e_2, \ldots \rangle$, in which each event $e_k$ is either the arrival or the departure of a DLT job $j$ with $w_j$ instances \\
\textbf{Output:}\ evolving state $\sigma$ updated after every event

\medskip
\noindent 1:\ $\sigma \leftarrow \sigma_0$ \\
\noindent 2:\ \textbf{for}\ each event $e$ in arrival order\ \textbf{do} \\
\noindent 3:\ \ \ \textbf{if}\ $e$ is the arrival of job $j$\ \textbf{then} \\
\noindent 4:\ \ \ \ \ \textbf{for}\ $k := 1, 2, \ldots, w_j$\ \textbf{do} \\
\noindent 5:\ \ \ \ \ \ \ $i \leftarrow$ the $k$-th instance of $j$ \\
\noindent 6:\ \ \ \ \ \ \ $(d, M) \leftarrow \mathsf{Place}(\sigma, i)$ \\
\noindent 7:\ \ \ \ \ \ \ install $i$ at $d$ and then apply $M$ to $\sigma$ \\
\noindent 8:\ \ \ \ \ \textbf{end for} \\
\noindent 9:\ \ \ \ \ $\sigma \leftarrow \sigma \oplus \mathsf{Compact}(\sigma)$ \\
\noindent 10:\ \ \ \textbf{else if}\ $e$ is the departure of job $j$\ \textbf{then} \\
\noindent 11:\ \ \ \ \ \textbf{for}\ $k := 1, 2, \ldots, w_j$\ \textbf{do} \\
\noindent 12:\ \ \ \ \ \ \ $i \leftarrow$ the $k$-th instance of $j$ \\
\noindent 13:\ \ \ \ \ \ \ $M \leftarrow \mathsf{Remove}(\sigma, i)$ \\
\noindent 14:\ \ \ \ \ \ \ remove $i$ from its current location and then apply $M$ to $\sigma$ \\
\noindent 15:\ \ \ \ \ \textbf{end for} \\
\noindent 16:\ \ \ \ \ $\sigma \leftarrow \sigma \oplus \mathsf{Compact}(\sigma)$ \\
\noindent 17:\ \ \ \textbf{end if} \\
\noindent 18:\ \textbf{end for} \\
\noindent 19:\ \textbf{return}\ $\sigma$
\end{minipage}}
\caption{The COMPASS algorithm: online event handler that composes the $\mathsf{Place}$, $\mathsf{Remove}$, and $\mathsf{Compact}$ operators to evolve the cluster state $\sigma$ inside the compact ABS $\Sigma$.}
\label{alg:compass}
\end{algorithm}

\subsection{Theoretical Guarantees}
\label{sec:guarantees}

We now state the main theoretical properties of the COMPASS-ABS scheduler. Proposition~\ref{prop:invariance} validates that COMPASS confines the cluster state to the ABS domain $\Sigma$ throughout. Theorem~\ref{thm:compactness} is the main claim of this section and shows that COMPASS-ABS keeps $\mathrm{SIF}^{\pi}(t)$ bounded by $2/N$ whenever the workload satisfies a realistic composition condition. Proposition~\ref{prop:complexity} characterises the per-operator computational complexity of the COMPASS algorithm. Detailed proofs of all three results are deferred to Appendix~\ref{app:sec:proofs}.

\begin{proposition}[State-Space Invariance]
\label{prop:invariance}
For any initial state $\sigma_0 \in \Sigma$ and any event sequence of job arrivals and job departures, every state produced by the COMPASS algorithm (Algorithm~\ref{alg:compass}) remains in $\Sigma$.
\end{proposition}

\begin{theorem}[Compactness under the Workload Composition Condition]
\label{thm:compactness}
Let $n_g(t)$ denote the number of active instances of per-GPU demand $g \in \{1,\ldots,8\}$ at time $t$, write $n_{\mathrm{total}}(t) = \sum_{g=1}^{8} n_g(t)$ for the total active instance count, and let $p_g(t) := n_g(t)/n_{\mathrm{total}}(t)$ denote the per-class instance share. Suppose that at every time $t$ the workload composition satisfies the \textbf{Workload Composition Condition}
\begin{equation}
p_1(t) \;\ge\; p_3(t) + 3\,p_5(t) + 2\,p_6(t) + p_7(t).
\label{eq:workload-cond-frac}
\end{equation}
Then in every cluster state $\sigma(t)$ produced by Algorithm~\ref{alg:compass} in response to the online event stream, at most one slotted node is non-fully-packed and at most one pool node is non-fully-packed, and consequently the online placement produced by COMPASS-ABS satisfies
\begin{equation}
\mathrm{SIF}^{\pi}(t) \;\le\; \frac{2}{N}, \qquad \forall\, t \ge 0
\label{eq:sif-compactness}
\end{equation}
where $N$ is the total number of nodes in the cluster.
\end{theorem}

\begin{proposition}[Operational Complexity of COMPASS]
\label{prop:complexity}
The three operators that compose COMPASS satisfy
\begin{equation}
\begin{aligned}
&\mathcal{T}(\mathsf{Place}) = O(\log N),\\
&\mathcal{T}(\mathsf{Remove}) = O(1), \\
&\mathcal{T}(\mathsf{Compact}) = O(N),
\end{aligned}
\label{eq:operator-complexity}
\end{equation}
where $N$ denotes the number of nodes in the cluster.
\end{proposition}

%% file: section5.tex
\section{Experimental Evaluation}
\label{sec:experiments}

\subsection{Testbed}
\label{sec:testbed}

We evaluate our scheduler in both a physical cluster and a simulation system. Our implementation is available at \url{https://anonymous.4open.science/r/COMPASSABSF02A/}.

For the physical experiments, our scheduler runs as a controller process
that launches and migrates real PyTorch training jobs on a small cluster of four servers each equipped with eight NVIDIA A100-40GB GPUs, and each migration is realized through a checkpoint-and-restart protocol on the actual training processes. 

For the simulation experiments, we develop an event-driven simulator in
Python to reproduce DLT job production in a shared GPU cluster. To faithfully reflect production behavior, our simulator reproduces the latency incurred by the
checkpoint-based migration mechanism\cite{xiao2018gandiva,gupta2024just}.

\subsection{Cluster Configuration}
\label{sec:cluster-config}

To validate our scheduler under both contended and loose resource regimes,
we replay real traces (or deploy realistic jobs) on a range of cluster
scales defined as follows. Let $\mathcal{J}_{[T_1, T_2]}$ represent the jobs submitted in a chosen window $[T_1, T_2]$, and let
\begin{equation}
D(t) \;=\; \sum_{j \in \mathcal{J}_{[T_1, T_2]}} G_j \,\mathbb{1}\!\left[a_j \le t < a_j + \tau_j\right]
\end{equation}
denote their concurrent GPU demand, with $G_j = W_j g_j$. We define the
window's $\eta$-th percentile demand $p_\eta$ as the smallest level $\ell$
that $D$ stays below for at least an $\eta$ fraction of its support. The
cluster for this window is then sized at
\begin{equation}
N(\eta) \;=\; \max\!\left(\left\lceil p_\eta / 8 \right\rceil,\; N^{\mathrm{maxjob}}\right)
\end{equation}
nodes, where  $N^{\mathrm{maxjob}}$
is the minimum node count required to host the largest single job in
$\mathcal{J}_{[T_1, T_2]}$ under the intra-node constraint and gang
scheduling requirement.

\subsection{Trace}
\label{sec:trace}

Our simulation is conducted on two real production traces, Venus\citep{hu2021characterization} and GFS\cite{duan2026gfs}. Both traces record the number of instances, per-instance GPU demand, submission time, and duration of each job. As an initial preprocessing step, we retain only the DLT jobs that consist of homogeneous instances with identical integer GPU demand. For the physical deployment, we synthesize a DLT trace (\emph{Phys-Mix}) spanning CV\cite{he2016resnet,dosovitskiy2021vit}, NLP\cite{devlin2019bert}, LLM\cite{radford2019gpt2}, recommender system\cite{guo2017deepfm}, and multi-modal\cite{radford2021clip} training tasks, whose resource demand characteristics are aligned with those observed in the two traces above.

\subsection{Comparison}
\label{sec:comparison}

We compare our scheduler against the six state-of-the-art methods.
\begin{itemize}
\item \textbf{GFS}\cite{duan2026gfs}: places each job's instances on
  the nodes whose remaining GPU capacity most closely matches the
  per-instance demand, following a best-fit policy.
\item \textbf{FGD}\cite{weng2023beware}: selects the placement that minimizes the resulting increase in fragmentation they define.
\item \textbf{CAFGD}\cite{lao2025cafgd}: combines its lifecycle-aware
  fragmentation measure with a wait-for-peak admission policy 
\item \textbf{MCG}\cite{wu2025mcgsched}: builds on FGD's fragmentation-gradient
  principle and additionally estimates the demand distribution of pending
  jobs to reorder scheduling priorities.
\item \textbf{FFT}\cite{mo2025fft}: globally reorganizes the placement of all running jobs by solving an ILP at the start of each round, and schedules arriving jobs within each round using a best-fit.
\item \textbf{DRR}\cite{wu2025defrag}: periodically migrates the jobs on the lowest-utilization node onto high-utilization target nodes at the start of each reschedule cycle, and places arriving jobs using a DRL policy trained via imitation learning from a best-fit heuristic.
\end{itemize}

\subsection{Metrics}
\label{sec:metrics}

We evaluate COMPASS-ABS and the baselines using three metrics that
together capture how effectively each scheduler mitigates the consequences
of GPU fragmentation. Among them, AFR (and the underlying SIF) is a
diagnostic quantity we introduce in this work to measure the cluster
fragmentation state; others are key performance indicators for
both users and vendors. Together they provide a complete assessment of
scheduler efficiency.

To focus on contention-induced behavior, we restrict the metrics indicating resource fragmentation degree and GPU utilization to the peak-demand period
\begin{equation}
T_{\mathrm{peak}} \;=\; \{\, t \in T \mid Q(t) \neq \emptyset \,\},
\end{equation}
that is, the time within the observed window $T$ when the waiting queue $Q(t)$ is non-empty.

\textbf{Average Fragmentation Rate (AFR).}
The time-averaged cluster fragmentation level during $T_{\mathrm{peak}}$, where the instantaneous fragmentation $\mathrm{SIF}(t)$ follows the
definition in \eqref{eq:sif}:
\begin{equation}
\mathrm{AFR} \;=\; \frac{1}{|T_{\mathrm{peak}}|} \int_{T_{\mathrm{peak}}} \mathrm{SIF}(t)\, dt.
\end{equation}
A lower AFR indicates that the scheduler keeps the cluster in a less
fragmented cluster state under contention.

\textbf{Average Job Completion Time (AJCT).}
The mean completion time across all DLT jobs submitted:
\begin{equation}
\mathrm{AJCT} \;=\; \frac{1}{M} \sum_{i=1}^{M} C_i,
\end{equation}
where $C_i$ is the completion time of job $J_i$, equivalently defined as the sum of its queueing time in the waiting queue, its training time on the assigned GPUs and latency caused by migration, and $M$ is the total number of jobs evaluated. A lower AJCT indicates that the scheduler admits and finishes jobs more
efficiently, which directly shortens the user-perceived turnaround of DLT workloads.

\textbf{Average GPU Utilization (AGU).}
The time-averaged fraction of active GPUs in the cluster during
$T_{\mathrm{peak}}$:
\begin{equation}
\mathrm{AGU} \;=\; \frac{1}{|T_{\mathrm{peak}}|} \int_{T_{\mathrm{peak}}} \mathrm{GU}(t)\, dt, \quad \mathrm{GU}(t) \in [0, 1].
\end{equation}
A higher AGU indicates that the scheduler converts more of the available
GPU capacity into productive work, which is critical under heavy GPU
demand.

%% file: section6.tex
\section{Evaluation}
\label{sec:evaluation}

\begin{figure*}[t]
  \centering
  \includegraphics[width=\textwidth]{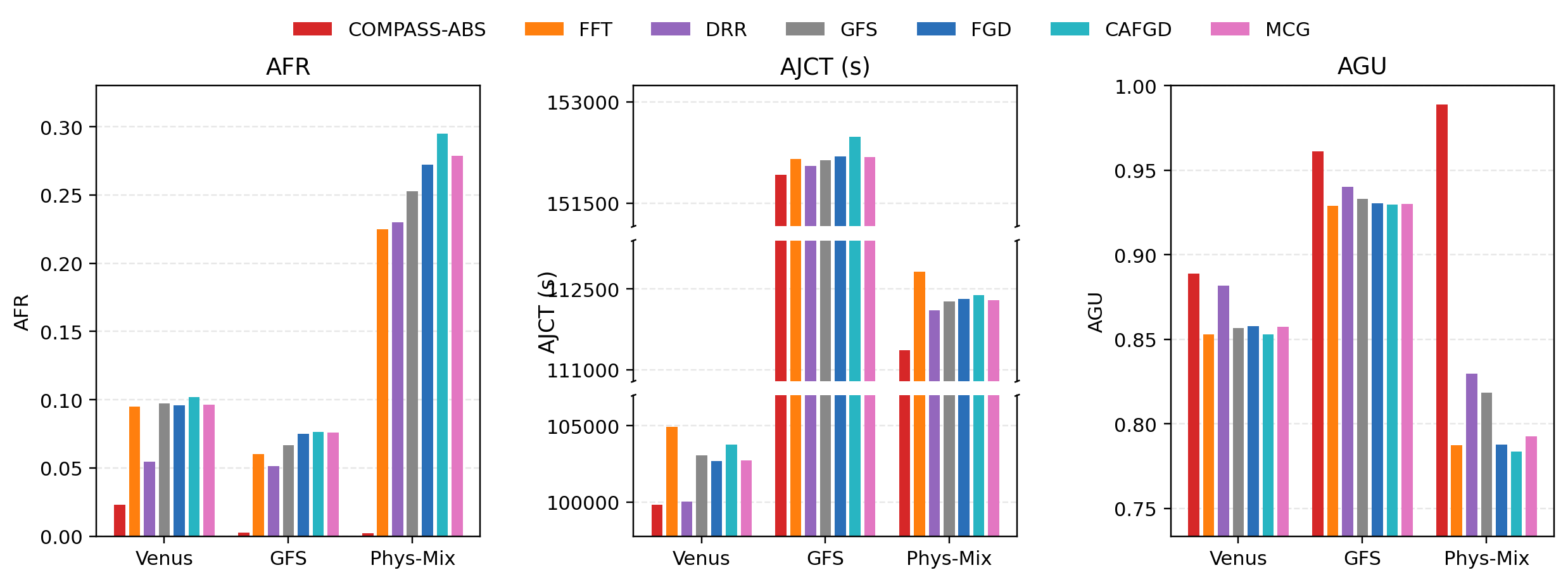}
  \caption{Overall comparison across three DLT workload traces ($\eta=0.90$).}
  \Description{Comparison of job completion time, GPU utilization, and fragmentation across the Venus, GFS, and Phys-Mix workload traces.}
  \label{fig:cross-dataset-comparison}
\end{figure*}

\subsection{Overall Comparison}
\label{sec:overall-comparison}

We first evaluate the end-to-end performance of COMPASS-ABS against all baselines under a fixed resource regime ($\eta = 0.90$). The comparison spans two real production traces replayed in our simulator and Phys-Mix deployed on our physical cluster, with results reported in Fig.~\ref{fig:cross-dataset-comparison}. Overall, COMPASS-ABS dominates the six baselines on every (metric, trace) pair. Concretely, COMPASS-ABS reduces AJCT by at least 212s, 135s and 736s and lifts AGU by at least $0.7\%$, $2\%$ and $16\%$ over the strongest baseline on the Venus, GFS and Phys-Mix traces respectively, while simultaneously attaining the lowest AFR on all three traces with the measured AFR staying below $\frac{2}{N}$, which matches the theoretical bound proved in Section~\ref{sec:guarantees}. COMPASS-ABS, DRR and FFT all support migration-driven defragmentation, yet ours exceeds DRR and FFT for a structural reason: our migration policy is driven by the COMPASS algorithm to maintain the compactness of the predefined feasible state domain ABS, and is tightly coupled with the placement strategy rather than treated as a separate phase at fixed interval.

\subsection{\texorpdfstring{Cluster Intensity Sensitivity}{Cluster Intensity Sensitivity}}
\label{sec:section-cluster-intensity-sensitivity}
\begin{figure*}[t]
  \centering
  \includegraphics[width=\textwidth]{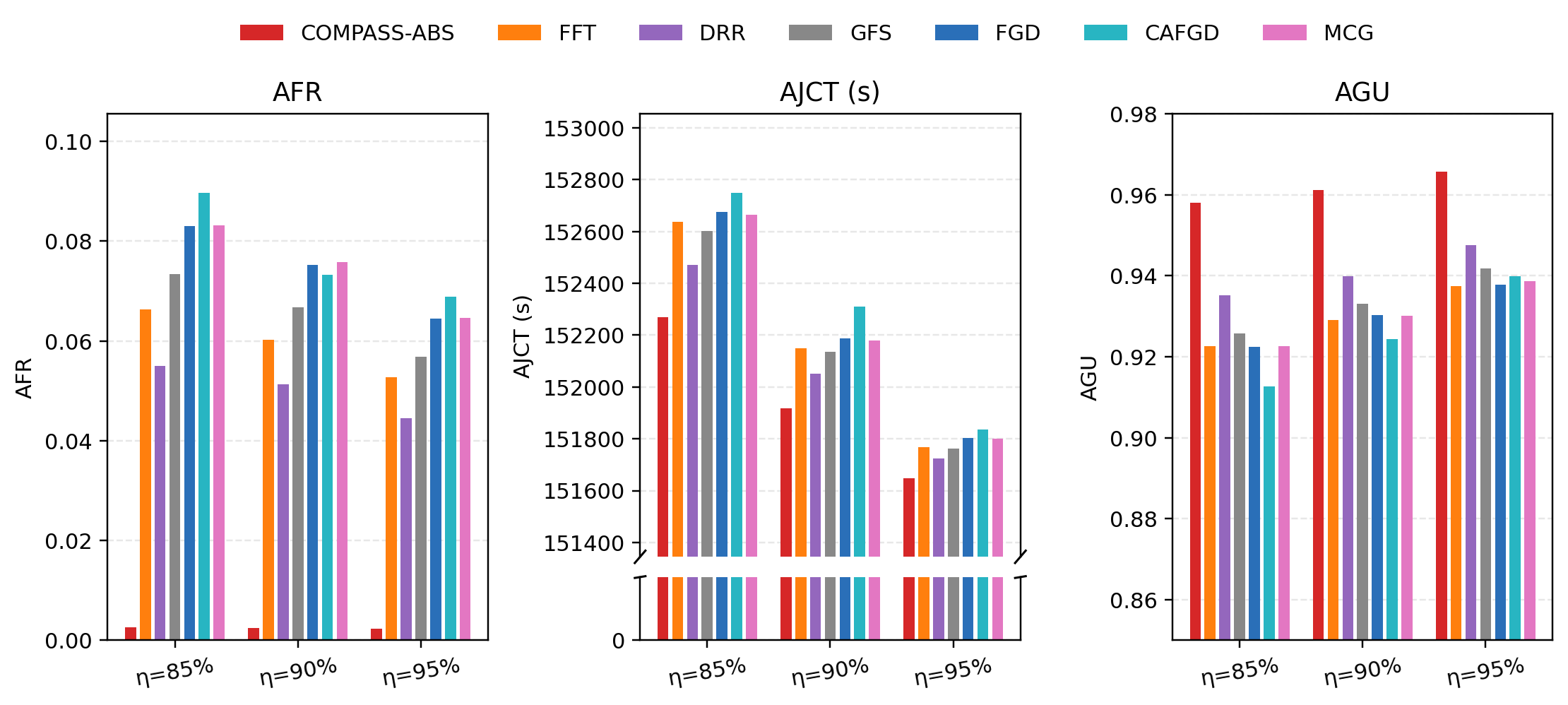}
  \caption{Sensitivity to cluster resource intensity on the GFS trace $\eta \in \{0.85,0.90,0.95\}$.}
  \label{fig:resource-intensity-robustness}
\end{figure*}

To verify that the superiority of COMPASS-ABS over the six baselines is stable across different cluster resource pressures, we sweep the cluster-scaling parameter $\eta \in \{0.85,\allowbreak  0.90,\allowbreak  0.95\}$ on the GFS trace, which carries the largest job count and aggregate GPU demand among the three traces and therefore offers the most discriminative regime for scheduler comparison. Across all three intensity levels, COMPASS-ABS dominates all baselines on every metric reported in Fig.~\ref{fig:resource-intensity-robustness}, with its AFR lying within the narrow band of 0.0023 to 0.0026 and its AGU staying above 0.957, and neither metric exhibits any appreciable sensitivity to resource pressure across the three regimes.

By contrast, the AJCT advantage of COMPASS-ABS over the strongest baseline widens monotonically as the cluster becomes tighter, growing from 74s at $\eta = 0.95$ to 135s at $\eta = 0.90$ and ultimately to 201s at $\eta = 0.85$. The mechanism behind is that fragmentation-induced queueing blockage grows increasingly severe as the cluster tightens: under tight resources, fragmented capacity is far less likely to be reclaimed by the natural departure of running jobs, whereas under loose resources such fragments are readily absorbed. COMPASS-ABS adaptively eliminates fragmentation in real time and thereby recovers this otherwise-trapped capacity, so its fragmentation-aware advantage translates into the largest AJCT reduction when the cluster is at its tightest.

\subsection{Robustness to Workload Composition}
\label{sec:robustness-to-workload-composition}

To verify that the superiority of COMPASS-ABS over the six baselines is also stable across different per-instance GPU-demand distributions, we select three days whose compositions span the widest pairwise discrepancy across our traces (Table~\ref{tab:workload-composition}). Across all three compositions, COMPASS-ABS dominates every baseline on every metric reported in Fig.~\ref{fig:workload-composition-robustness}: its AFR is confined within $[0.002, 0.030]$ against the baselines' wider $[0.058, 0.324]$ band, its AGU stays above $0.89$ against the baselines' $[0.77, 0.90]$ range, and its AJCT leads the strongest baseline by 1{,}643~s on Venus day~92, 1{,}431~s on GFS day~58 and 145~s on GFS day~26.

This robustness arises from the fact that the slot-pool layout exploits the dominance of the fragmentation-friendly $\{2, 4, 8\}$-GPU widths over the fragmentation-unfriendly $\{3, 5,$ $ 6, 7\}$-GPU widths in real DLT workloads, so the anchor-filler mechanism saturates the slots and the COMPASS algorithm confines the cluster state within a compact ABS configuration regardless of how the per-instance GPU-demand distribution shifts. The absolute AJCT lead is the largest on Venus day~92 because most of the JCT on this short-job workload is determined by scheduling decisions, whereas it shrinks on the two GFS days because their longer base runtimes dominate the JCT and leave a smaller scheduling-controllable share for any scheduler to optimize; even so, COMPASS-ABS still ranks first against other baselines.
\begin{table}[h]
  \caption{Workload Composition of the Three Selected Days}
  \label{tab:workload-composition}
  \small
  \setlength{\tabcolsep}{4pt}
  \begin{tabular}{lccc}
    \toprule
    & \multicolumn{3}{c}{Instance Size Ratio} \\
    \cmidrule(lr){2-4}
    Day          & 1 GPU & 2--7 GPU & 8 GPU \\
    \midrule
    Venus Day 92 & 0.73  & 0.06     & 0.21  \\
    GFS Day 58   & 0.57  & 0.04     & 0.39  \\
    GFS Day 26   & 0.43  & 0.09     & 0.48  \\
    \bottomrule
  \end{tabular}
\end{table}

\begin{figure*}[t]
  \centering
  \includegraphics[width=\textwidth]{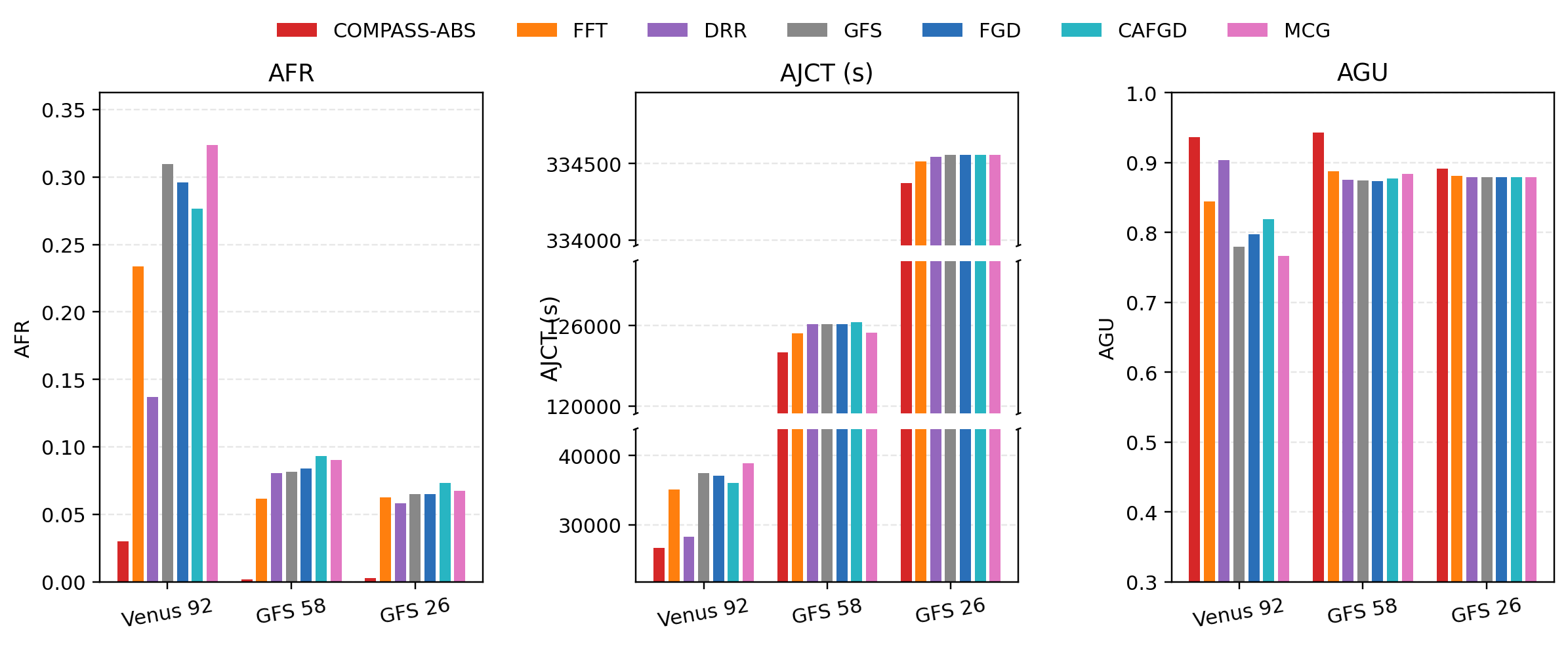}
  \caption{Robustness to workload composition on three representative days (Venus Day~92, GFS Day~58, GFS Day~26)}.
  \label{fig:workload-composition-robustness}
\end{figure*}

\subsection{Workload Composition Condition: Coverage and Robustness on Violating Days}
\label{sec:wcc-coverage}

The performance guarantee of COMPASS-ABS rests on Theorem~\ref{thm:compactness}, which assumes the cluster state satisfies the Workload Composition Condition (WCC) at every scheduling tick. We therefore test how often the WCC holds and whether COMPASS-ABS keeps its advantage on days where it does not hold throughout. We replay every Venus and GFS day, measure each day's WCC coverage $\rho$, the fraction of its scheduling horizon during which the WCC is satisfied, bucket the days by $\rho$, and for each bucket and each metric in $\{\mathrm{AFR}, \mathrm{AJCT}, \allowbreak\mathrm{AGU}\}$ count the days on which COMPASS-ABS ranks top-1 among the seven evaluated schedulers (Table~\ref{tab:wcc-top1-ranking}).

The WCC is empirically dominant on both traces. On Venus, 487 of the 564 replayed days (86.3\%) satisfy it throughout the entire horizon, and a further 51 days (9.0\%) satisfy it for at least 90\% of the horizon; on GFS it holds across the full horizon on all 552 days. The WCC-satisfying regime is thus the principal operating regime of production deployments rather than a corner case.

Even on the Venus days where the WCC does not hold throughout, the top-1 frequency of COMPASS-ABS closely tracks that on the $\rho = 1$ days. Across the five coverage buckets with $\rho < 1$, it leads on AFR for 100\% of the days in the three lower-coverage buckets $[0, 0.1)$, $[0.1, 0.7)$, and $[0.7, 0.9)$ and for at least 88\% in the two higher-coverage buckets, while on AJCT and AGU it leads for at least 64\% in every WCC-violating bucket and for 100\% in the lowest-coverage bucket $[0, 0.1)$. Aggregated across all 564 Venus days, the top-1 frequency reaches 94.1\% on AFR, 80.3\% on AJCT, and 81.9\% on AGU, and on GFS 100\%, 97.3\%, and 95.9\% respectively. COMPASS-ABS thus retains its advantage on all three metrics even outside this regime.

\subsection{Robustness to Migration Cost}
\label{sec:robustness-to-migration-cost}

Empirical migration latencies in DLT clusters span the range of
1s to 10s, depending on the checkpoint-and-restore implementation,
the network topology and the model
size~\cite{xiao2018gandiva,mo2025fft,ye2024survey,han2020marble}; on our own physical cluster the average migration latency we measure is approximately 8s. We
therefore verify that COMPASS-ABS maintains its superiority across the realistic range by sweeping the per-migration cost parameter
$c \in \{1, \ldots, 10\}$~s on the GFS and Venus traces.

\begin{figure}
  \centering
  \includegraphics[width=\columnwidth]{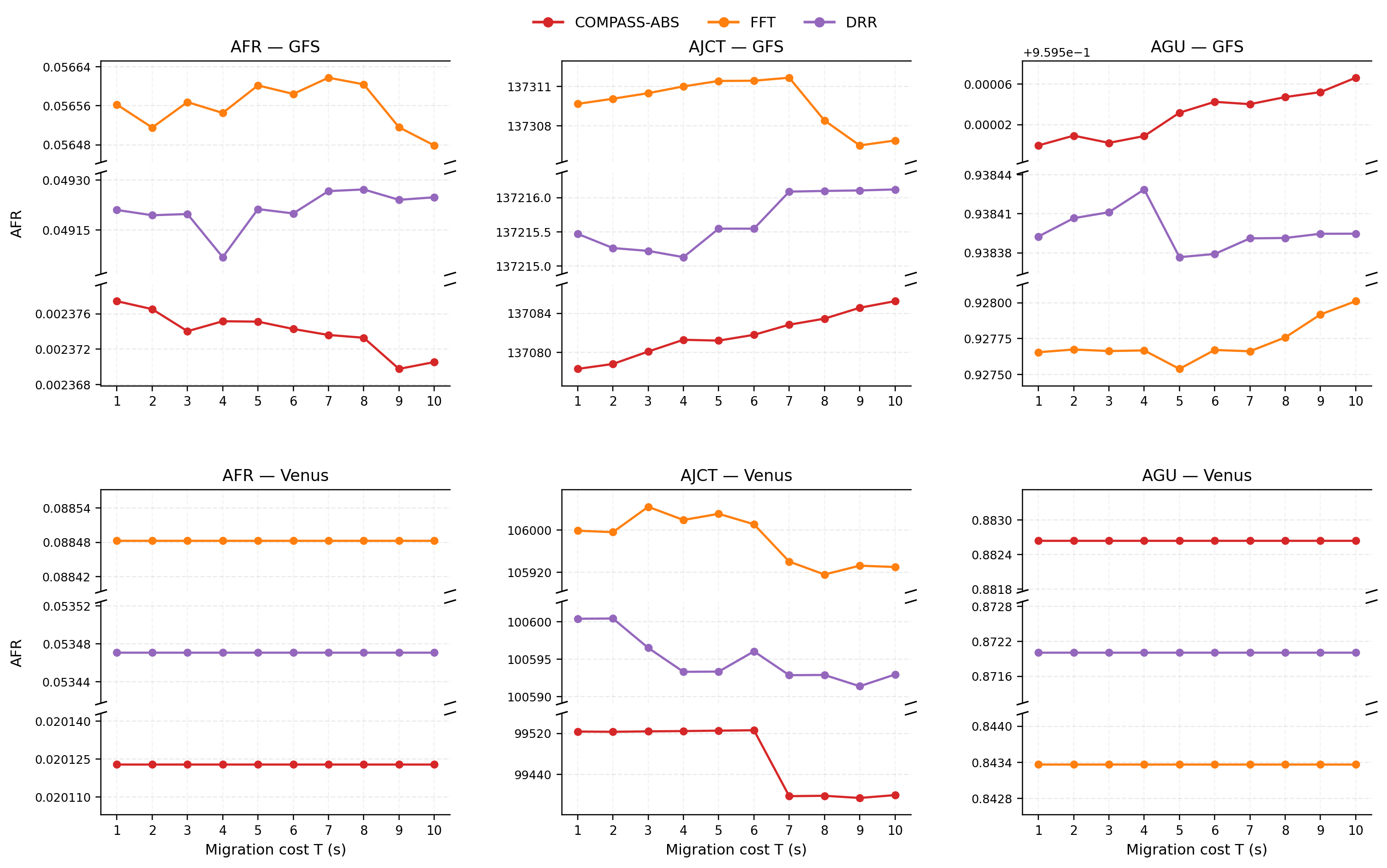}
  \caption{Robustness to per-migration cost on the GFS and Venus trace,
$c \in [1, 10]$~s, $\eta = 0.90$.}
  \label{fig:migration-cost-robustness}
\end{figure}

As shown in Fig.~\ref{fig:migration-cost-robustness}, every one of the six panels exhibits essentially flat curves as the per-migration cost $c$ sweeps from 1~s to 10~s, because the three migration-capable schedulers each issue at most a small number of migrations per job and the cumulative cost contribution to any metric therefore stays well below one second per job even at $c = 10$~s. More importantly, COMPASS-ABS attains the lowest AFR, the lowest AJCT and the highest AGU at every $c$ in the swept range on both GFS and Venus, including the 8~s neighborhood that corresponds to the latency we measure on our own cluster, so the end-to-end advantage of COMPASS-ABS is insulated against any realistic increase in the per-migration cost.

\begin{table}
  \caption{Top-1 ranking frequency of COMPASS-ABS across WCC coverage $\rho$ buckets, grouped by $(\text{day},\eta)$
  pairs.}
  \label{tab:wcc-top1-ranking}
  \footnotesize
  \setlength{\tabcolsep}{3pt}
  \renewcommand{\arraystretch}{1.0}
  \begin{tabular}{llrccc}
    \toprule
    Trace & $\rho$ Bucket & $|D|$ & AFR & JCT & AGU \\
    \midrule
    \multirow{16}{*}{Venus}
      & \multirow{2}{*}{$1.0$}                  & \multirow{2}{*}{487} & 459      & 383      & 404      \\
      &                                         &                      & (94.3\%) & (78.6\%) & (83.0\%) \\
      & \multirow{2}{*}{$[0.99,\,1)$}           & \multirow{2}{*}{26}  & 24       & 18       & 20       \\
      &                                         &                      & (92.3\%) & (69.3\%) & (76.9\%) \\
      & \multirow{2}{*}{$[0.9,\,0.99)$}         & \multirow{2}{*}{25}  & 22       & 16       & 16       \\
      &                                         &                      & (88.0\%) & (64.0\%) & (64.0\%) \\
      & \multirow{2}{*}{$[0.7,\,0.9)$}          & \multirow{2}{*}{11}  & 11       & 8        & 9        \\
      &                                         &                      & (100\%)  & (72.7\%) & (81.8\%) \\
      & \multirow{2}{*}{$[0.1,\,0.7)$}          & \multirow{2}{*}{10}  & 10       & 8        & 8        \\
      &                                         &                      & (100\%)  & (80.0\%) & (80.0\%) \\
      & \multirow{2}{*}{$[0,\,0.1)$}            & \multirow{2}{*}{5}   & 5        & 5        & 5        \\
      &                                         &                      & (100\%)  & (100\%)  & (100\%)  \\
      & \multirow{2}{*}{\textbf{Total}}         & \multirow{2}{*}{\textbf{564}} & \textbf{531}      & \textbf{453}      & \textbf{462}      \\
      &                                         &                      & \textbf{(94.1\%)} & \textbf{(80.3\%)} & \textbf{(81.9\%)} \\
    \midrule
    \multirow{2}{*}{GFS}
      & \multirow{2}{*}{$1.0$}                  & \multirow{2}{*}{552} & 552      & 537      & 529      \\
      &                                         &                      & (100\%)  & (97.3\%) & (95.9\%) \\
    \bottomrule
  \end{tabular}
\end{table}

\subsection{Case Study of Fragmentation Dynamics}
\label{sec:case-studies}

We zoom into the 21:00--22:30 window of GFS day 113 under $\eta = 0.95$ to expose the time-resolved fragmentation mechanism, with Fig.~\ref{fig:case-sif}, Fig.~\ref{fig:case-util} and Fig.~\ref{fig:case-idle} reporting the dynamic $N \cdot \mathrm{SIF}$ (the total number of avoidable partial nodes), GPU utilisation and fully-empty node counts respectively. All six practical baselines remain in continuous blocking throughout the window, with five of them attributing $100\%$ of the blocking to fragmentation and DRR still $94.6\%$, as verified by the aggregate idle GPU count strictly exceeding the head-of-line demand. Over $90\%$ of the blocked jobs carry $g_i = 8$ with $w_i > 30$, so the dominant blocking factor is the requirement for at least $w_i$ fully empty 8-GPU nodes.

On the fragmentation side, the five non-migration baselines settle into a wide 60--80 partial-node band while FFT confines itself to a 25--50 band through periodic global migration, but both ranges leave only 20--30 fully empty nodes against the $w_i > 30$ contiguous demand and trap GPU utilisation within 0.8--0.9 for the entire window. COMPASS-ABS keeps $N \cdot \mathrm{SIF}$ bounded by 2 throughout, so its total blocking accumulates to only about 10 minutes, and although $72.9\%$ of that blocking is technically fragmentation-induced, the compact ABS layout together with its slot-pool migration mechanism dispatches every such blockage within seconds and never lets fragmentation translate into long-period queue blocking or low utilisation.

\begin{figure}
  \centering
  \includegraphics[width=\columnwidth]{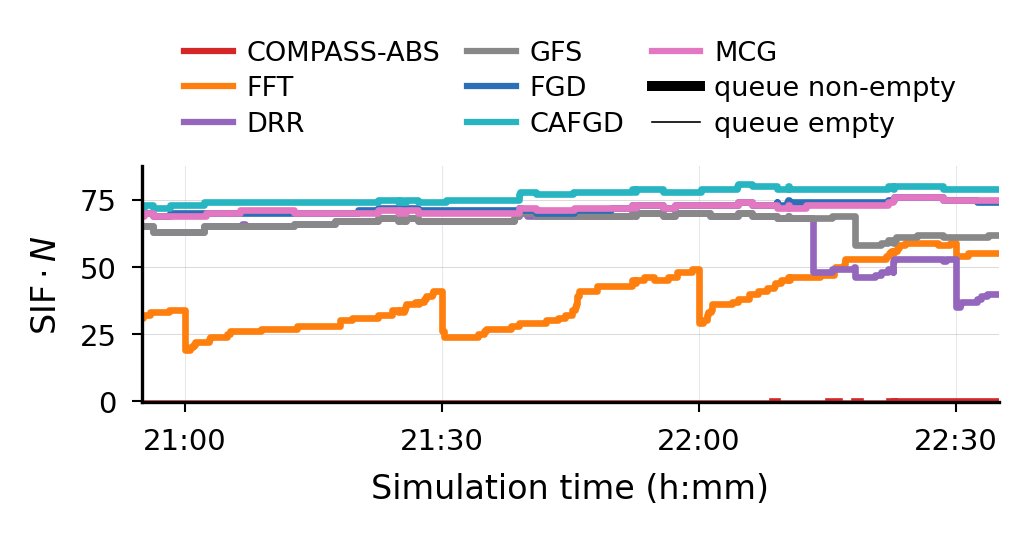}
  \caption{Dynamic number of avoidable partial nodes ($N \cdot \mathrm{SIF}$) by seven schedulers on GFS day 113 ($\eta = 0.95$).}
  \label{fig:case-sif}
\end{figure}

\begin{figure}[b]
  \centering
  \includegraphics[width=\columnwidth]{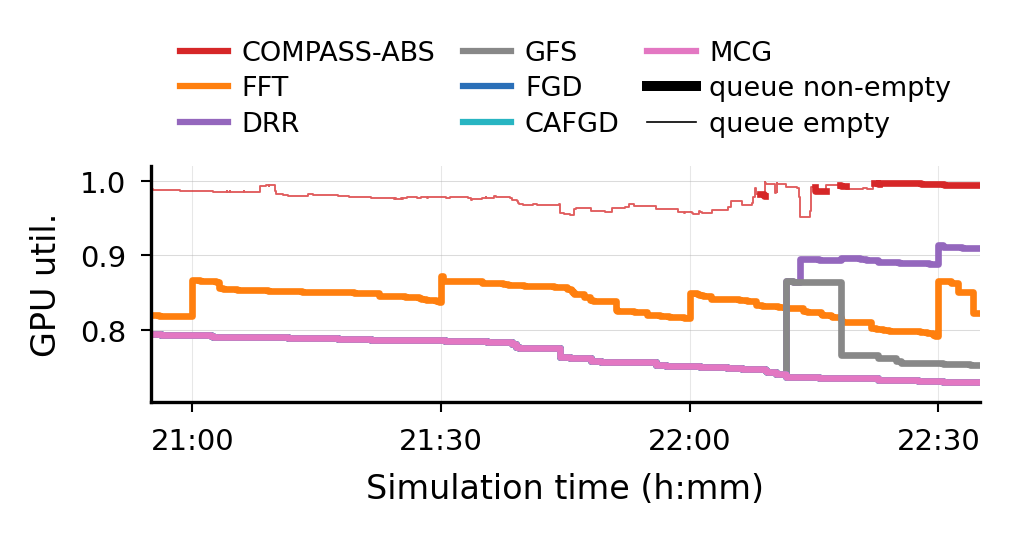}
  \caption{Dynamic GPU utilization by seven schedulers on GFS day 113 ($\eta = 0.95$).}
  \label{fig:case-util}
\end{figure}

\begin{figure}
  \centering
  \includegraphics[width=\columnwidth]{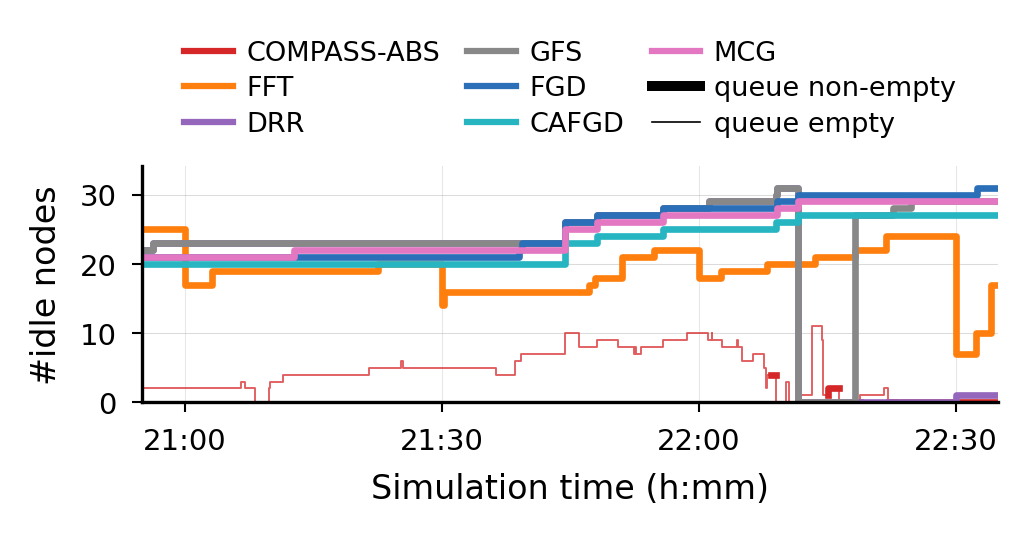}
  \caption{Dynamic free nodes by seven schedulers on GFS day 113 ($\eta = 0.95$).}
  \label{fig:case-idle}
\end{figure}

%% file: section7.tex
\section{Related Work}

Most prior work on DLT job scheduling overlooks the quantification of resource fragmentation, asserting that they mitigate fragmentation via tight packing, without providing a concrete metric\cite{li2023lyra, wu2023tgs, rajasekaran2024cassini}. FGD\cite{weng2023beware} was the first to propose a statistical measure that estimates the expected instance capacity that the remaining resources can accommodate, given prior knowledge of the resource demand distribution. Subsequent works refined this formulation on top of FGD: CAFGD\cite{lao2025cafgd} decomposed the demand probability into weighted long-term and short-term distributions, while MCG\cite{wu2025mcgsched} incorporated load balance into the fragmentation metric. However, all of these approaches rely on historical workload information.

Numerous studies aimed to reduce DLT job completion time and improve resource utilization by mitigating fragmentation. Existing schedulers can be broadly categorized into two types of approaches. Non-migration schedulers: These methods formulated scheduling as a multi-dimensional bin-packing problem\cite{verma2014evaluating,grandl2014multiresource,huang2015multiresource,zhao2022multiresource,chen2023deepboot}. ElasticFlow\cite{gu2023elasticflow}, Synergy\cite{mohan2022synergy}, and GFS\cite{duan2026gfs} adopted a best-fit strategy that dynamically places jobs on the node leaving the fewest idle GPUs. FGD\cite{weng2023beware}, CAFGD\cite{lao2025cafgd}, and MCG\cite{wu2025mcgsched} scheduled jobs based on their respective fragmentation measures, preferring placements that minimize the increase between pre- and post-scheduling fragmentation. These methods cannot reduce fragmentation that has already accumulated in the cluster. Migration-supporting schedulers: To address accumulated fragmentation, schedulers such as Gavel\cite{narayanan2020gavel}, Sia\cite{jayaram2023sia}, and RASA\cite{chen2024rasa} modeled the problem as a linear program computed by numerical solvers. However, the resulting time complexity and additional latency induced by frequent migrations are prohibitive for real-time scheduling in large-scale clusters\cite{woeginger1997noptas}. To improve practicality, Hops\cite{wang2024hops}, FFT\cite{mo2025fft} and DRR\cite{wu2025defrag} employed a round-based migration strategy that performs global defragmentation by relocating running jobs to placements that incur less fragmentation. Nevertheless, these methods cannot correct fragmentation introduced within a round, leaving the cluster in a degraded state until the next defragmentation cycle.

%% file: section8.tex
\section{Conclusion}

In this paper, we have presented Scheduler-Induced Fragmentation (SIF) for quantifying resource fragmentation together with COMPASS-ABS, a scheduler that reduces fragmentation in shared GPU clusters for DLT jobs. From the partial-nodes perspective, SIF removes the dependence on prior knowledge of the workload trace that earlier statistical metrics require, and isolates the portion of observed fragmentation resulting from scheduling policy. Building on this measure, COMPASS-ABS confines the cluster state to a structured feasible domain called the Anchor-Based Space (ABS), whose construction fully exploits the alignment between the GPU-demand profile of DLT instances and the cluster topology, so that the accompanying COMPASS algorithm can dynamically maintain compactness inside ABS and keep SIF uniformly bounded by $2/N$ whenever the Workload Composition Condition holds. We validate our scheduler's performance through large-scale simulation on production traces and through physical-cluster deployment with a synthetic DLT workload, both of which show that COMPASS-ABS attains higher GPU utilization and substantially shorter job completion time in the queue than the state-of-the-art baselines even when WCC is violated.

Looking ahead, we identify three promising directions for extending COMPASS-ABS. First, we plan to generalize our scheduler to jointly handle both deep learning training and inference jobs in a shared GPU cluster, where the heterogeneous latency requirements and resource profiles of the two workload types pose additional scheduling challenges. Second, when migrating jobs to amend fragmentation, the current design treats all jobs uniformly; a natural extension is to incorporate job-level priorities into the migration decision, selecting victims in a manner that respects Service Level Objectives (SLOs) and avoids penalizing high-priority workloads. Moreover, we further discuss how to extend our scheduler to the next-generation clusters with larger NVLink domains in Appendix~\ref{app:sec:nvl}.

%% file: Proof/main_proofs.tex
\section{Compactness of Best-Fit on Power-of-Two Workloads}
\label{app:sec:bf-power-of-two}

This appendix formalises Observation~1 of Section~\ref{sec:design-rationale}.
We consider the restricted setting in which every multi-GPU instance has
demand $g_i \in \{2, 4, 8\}$, i.e.\ the active workload is contained in
$\mathcal{I}^{\mathrm{FF}}$, and the cluster runs a stripped-down variant
of Algorithm~\ref{alg:compass}: each arrival is handled by a single
$\mathsf{Place}$ call that uses best-fit (no $\mathsf{Compact}$ pass on
arrival), and each departure is handled by per-instance $\mathsf{Remove}$
calls followed by a single $\mathsf{Compact}$ pass on the departure side
only. We show that this minimal policy is already sufficient to enforce
$\mathrm{SIF}^{\pi}(t) \le 2/N$ at all times, and we exhibit a finite
event sequence that reaches the bound.

\subsection{Setup and statement}
\label{app:subsec:bf-setup}

Let the active workload consist of items of size $g \in \{2, 4, 8\}$ that
arrive and depart online, and let each node be a bin of capacity $G = 8$.
We use $\Phi^{\pi}(t) := \nu(P^{\pi}(t))$ for the partial-node count
under policy $\pi$ at time $t$, with $F^{\pi}(t) = \Phi^{\pi}(t)/N$ as in
Section~\ref{sec:decomposing}. The free space of a partial bin lies in
$\{2, 4, 6\}$, because used capacity attainable by sums of items in
$\{2,4,8\}$ that strictly under-fill $G = 8$ is one of $\{2, 4, 6\}$.

\begin{proposition}[Best-Fit Compactness on Power-of-Two Workloads]
\label{prop:bf-power-of-two}
Let the policy place every arriving instance via best-fit, i.e.\ on the
node leaving the smallest remaining capacity after placement and opening
a fresh node when no existing node can host the instance, and let every
departure be followed by a single $\mathsf{Compact}$ pass. If every
active instance has $g_i \in \{2, 4, 8\}$, then for every $t \ge 0$ the
state $\sigma(t)$ produced by the policy satisfies
\[
  \Phi^{\pi}(t) \;\le\; 2 \qquad \text{and} \qquad
  \mathrm{SIF}^{\pi}(t) \;\le\; \frac{2}{N},
\]
and both inequalities are tight in the sense that a finite arrival
sequence reaches $\Phi^{\pi}(t) = 2$.
\end{proposition}

\subsection{Proof of Proposition~\ref{prop:bf-power-of-two}}
\label{app:subsec:bf-proof}

We bound $\Phi^{\pi}(t)$ separately on the two event types and combine
the bounds. Throughout, ``bin'' is interchangeable with ``node'' and
``item'' with ``instance''.

\paragraph{Part 1: best-fit on arrival leaves at most two partial bins.}
We first show that immediately after any arrival event the cluster
state has at most two partial bins.

\subparagraph*{Lemma B.1 (free-value uniqueness).}
At any time during the arrival pass, at most one partial bin has
remaining free space $r$ for each $r \in \{2, 4, 6\}$.

\begin{proof}
Suppose for contradiction that two partial bins $A$ and $B$ have free
space $r$ at the same time. Let $B$ be the bin whose free space last
became $r$, and consider the placement event that produced
$\mathrm{free}(B) = r$. Just before that placement, $\mathrm{free}(A) = r$
already held by hypothesis. The placement put an item of size $s$ into
$B$, with $s \le \mathrm{free}(B)_{\mathrm{before}}$ and
$\mathrm{free}(B)_{\mathrm{before}} - s = r$, so
$\mathrm{free}(B)_{\mathrm{before}} \ge s$. If $B$ was a fresh bin, then
$\mathrm{free}(B)_{\mathrm{before}} = G = 8$ and $s = G - r > 0$; bin
$A$ also had free space $r \ge s$, since $r = G - s$ and $G \ge s$
forces $r \ge 0$, and we ruled out $r = 0$ because $A$ is partial.
Concretely, for every reachable $r \in \{2, 4, 6\}$ the item size
$s = G - r \in \{6, 4, 2\}$, but only $s \in \{2, 4\}$ are admissible
power-of-two items (size $6$ is not in the workload), so the only case
is $s \in \{2, 4\}$, in which $A$ with $\mathrm{free}(A) = r$ also has
$\mathrm{free}(A) \ge s$ and best-fit would have picked $A$ over the
fresh $B$ because the fresh bin has $G = 8 > r$.  If $B$ was not fresh, then $\mathrm{free}(B)_{\mathrm{before}} > r$, so
$\mathrm{free}(A) = r < \mathrm{free}(B)_{\mathrm{before}}$. Best-fit's
least-fit rule would then still have preferred $A$ over $B$, because $A$
leaves the smaller free space, provided $A$ can host the item. The
admissibility condition $\mathrm{free}(A) \ge s$ is exactly $r \ge s$. It
holds because the placement into $B$ required
$\mathrm{free}(B)_{\mathrm{before}} \ge s$, and substituting
$r = \mathrm{free}(B)_{\mathrm{before}} - s$ rewrites $r \ge s$ as
$\mathrm{free}(B)_{\mathrm{before}} \ge 2 s$.
  The only case where
$\mathrm{free}(A) < s$ is when $r < s$, but this is already excluded by
the construction $\mathrm{free}(B)_{\mathrm{before}} = r + s$ and the
item size $s$ being one of $\{2, 4\}$ which gives $r + s \le G$ hence
$r \le G - s$. In every admissible sub-case best-fit's tie-breaking
prefers $A$, contradicting the placement having gone to $B$.
\end{proof}

\subparagraph*{Lemma B.2 (free-six isolation).}
If some partial bin has free space $6$, then no other partial bin
exists at the same time.

\begin{proof}
A bin with $\mathrm{free} = 6$ contains exactly one size-$2$ item.
Consider the most recent placement that put this size-$2$ item into the
bin. The bin was empty just before the placement, because no
other item-size combination from $\{2, 4, 8\}$ yields used capacity
$2$. Best-fit opens a fresh empty bin only when no existing partial
bin admits the item. The arriving item has size $2$, so an admissible
partial bin would only require $\mathrm{free} \ge 2$, which every
partial bin satisfies because every partial free value lies in
$\{2, 4, 6\}$. Hence no partial bin existed at the moment the
size-$2$ item was placed in the fresh bin, and the new $\mathrm{free} =
6$ bin is the sole partial bin at that instant.

Each subsequent arrival event either (i) places a size-$2$ item, which
best-fit routes into the existing $\mathrm{free} = 6$ bin and reduces
its free space to $4$, removing it from the $\mathrm{free} = 6$ class;
or (ii) places a size-$4$ item, which best-fit routes into the existing
$\mathrm{free} = 6$ bin and reduces its free space to $2$; or (iii)
places a size-$8$ item, which opens a fresh bin that immediately
becomes full and creates no new partial bin. In none of these cases is
a second partial bin created while the original $\mathrm{free} = 6$ bin
still has free space $6$, so the isolation invariant is preserved
until the next departure event.
\end{proof}

\subparagraph*{Lemma B.3 (two-partial bound).}
At any time immediately after an arrival placement, the cluster has at
most two partial bins, and the only reachable two-partial configuration
has free spaces $\{2, 4\}$.

\begin{proof}
Combining Lemmas~B.1 and~B.2, the set of free-space values present
across partial bins is a subset of $\{2, 4, 6\}$ in which each value
appears at most once, and the value $6$ cannot coexist with any other
partial bin. Hence the only feasible multisets of free-space values
across partial bins are $\emptyset$, $\{2\}$, $\{4\}$, $\{6\}$,
$\{2, 4\}$, $\{2, 6\}$, $\{4, 6\}$ from cardinality counting, and the
last two are forbidden by Lemma~B.2. The remaining configurations have
at most two partial bins.
\end{proof}

\paragraph{Part 2: Compact on departure leaves at most one partial bin.}

Departures invoke a single $\mathsf{Compact}$ pass after the per-instance
$\mathsf{Remove}$ calls have been applied. Under the present restriction
$\mathcal{I}^{=1} = \emptyset$, $\mathcal{I}^{\mathrm{FUF}} = \emptyset$,
so the active 1-GPU population $n_1(t)$ is identically zero and the
total slot vacancy aggregated across anchored fragmentation-unfriendly
instances is identically zero. The Workload Composition Condition
\eqref{eq:workload-cond-frac} therefore reduces to $0 \ge 0$ and holds
trivially at every time. Lemma~A.3 of
Appendix~\ref{app:sec:proofs} then contributes zero non-fully-packed
pool nodes because no pool nodes exist, and Lemma~A.4 of
Appendix~\ref{app:sec:proofs} contributes at most one non-fully-packed
slotted node from the width-$4$ and width-$2$ $\mathsf{BestFitDrain}$
passes. The post-departure state therefore satisfies $\Phi^{\pi}(t)
\le 1$.

\paragraph{Combining the two parts.}

The cluster state changes only at event boundaries. After every
arrival event Part~1 gives $\Phi^{\pi}(t) \le 2$, and after every
departure event Part~2 gives $\Phi^{\pi}(t) \le 1 \le 2$. The bound
$\Phi^{\pi}(t) \le 2$ therefore holds at every event boundary and
extends to every continuous time instant between consecutive events,
so
\[
  \begin{aligned}
    &\Phi^{\pi}(t) \;\le\; 2, \\
    &F^{\pi}(t) \;=\; \frac{\Phi^{\pi}(t)}{N} \;\le\; \frac{2}{N}, \\
    &\mathrm{SIF}^{\pi}(t) \;\le\; F^{\pi}(t) \;\le\; \frac{2}{N}.
  \end{aligned}
\]
for all $t \ge 0$, where the last inequality uses the non-negativity
of the inherent fragmentation $F^{\star}(t)$ from
Section~\ref{sec:decomposing}.

\paragraph{Tightness.}

The bound $\Phi^{\pi}(t) = 2$ is realised by the following arrival
sequence on an otherwise fully-packed cluster. Suppose all nodes other
than one are fully packed and the remaining node hosts a single size-$4$
anchor together with one size-$2$ anchor, leaving free space $2$. A
DLT job composed of three size-$4$ instances now arrives. Best-fit
considers each arriving instance in turn. The first instance encounters
free space $2$ on the partial node, which does not admit a size-$4$
item, so a fresh node is opened and the first instance is placed there,
producing a new partial node with free space $4$. The second instance
sees free space $2$ on the original partial node and free space $4$
on the just-opened node, prefers the smaller fitting capacity, and is
placed onto the just-opened node, which becomes fully packed. The third
instance again finds the original partial node uninhabitable for a
size-$4$ item and the now-full node also unable to host it, so a
further fresh node is opened with free space $4$. The cluster now
contains two partial nodes, one with free space $2$ and one with free
space $4$, matching the only feasible two-partial configuration
identified in Lemma~B.3 and witnessing $\Phi^{\pi}(t) = 2$ together
with $\mathrm{SIF}^{\pi}(t) = 2/N$. This completes the proof of Proposition~\ref{prop:bf-power-of-two}.
\qed

\section{Proofs for the Theoretical Guarantees of COMPASS-ABS}
\label{app:sec:proofs}

This appendix collects the detailed proofs of the three results stated in
Section~\ref{sec:guarantees}, namely the unconditional state-space
invariance in Proposition~\ref{prop:invariance}, the conditional
$\mathrm{SIF} \le 2/N$ compactness bound in Theorem~\ref{thm:compactness},
and the per-operator complexity in Proposition~\ref{prop:complexity}. We
keep the notation introduced in Section~\ref{sec:state-space} for the
ABS domain $\Sigma$ and in Section~\ref{sec:scheduling-policy} for the
three operators $\mathsf{Place}$, $\mathsf{Remove}$, and $\mathsf{Compact}$.

\subsection{Proof of Proposition~\ref{prop:invariance} (State-Space Invariance)}
\label{app:subsec:invariance}

We prove that every event applied by Algorithm~\ref{alg:compass} maps
$\Sigma$ into $\Sigma$, so that by induction every state generated from
an initial $\sigma_0 \in \Sigma$ remains in $\Sigma$. Since the algorithm
decomposes each job-level event into a finite sequence of per-instance
operator calls together with at most one $\mathsf{Compact}$ call, it
suffices to show that each operator branch preserves $\Sigma$ when
applied to a state already in $\Sigma$.

\paragraph{Place preserves $\Sigma$.}
Consider an arrival $i$ with demand $g_i$ applied to $\sigma \in \Sigma$.
We argue on the branches of $\mathsf{Place}$.

If $g_i \ge 2$, the operator selects a width $w = w^\star(g_i) \in
\{2,4,8\}$ and a host $n^\star$ such that either $n^\star$ is already
slotted with $G - \sum_{s \in S(n^\star)} w(s) \ge w$, or $n^\star$ was
empty and is promoted to slotted with an empty profile. In the first
case, appending a slot of width $w$ to $S(n^\star)$ produces a new slot
multiset whose total width is at most $G$ and whose constituents lie in
$\{2,4,8\}$, hence the new $S(n^\star)$ remains in $\mathsf{SlotConf}$
as defined in \eqref{eq:slotconf}. In the second case, the post-event
$S(n^\star) = \{w\}$ trivially lies in $\mathsf{SlotConf}$. The anchor of
the new slot is set to $i$ so that $w(s^\star) = w^\star(g_i) =
w^\star(g_{a(s^\star)})$, which matches the constraint of
Definition~\ref{def:anchor}. The associated filler set is initialised to
the up-to-$r$ migrated 1-GPU instances with $r = w - g_i$, so that
$|F(s^\star)| \le r = w(s^\star) - g_{a(s^\star)}$ and the capacity
invariant \eqref{eq:capacity-invariant} holds as an inequality. Every
pool node from which a filler is drawn loses one occupant and therefore
preserves $|P(n)| \le G$. No other node is touched.

If $g_i = 1$, the operator places $i$ either into a slot vacancy
$(g_{a(s)} + |F(s)| < w(s))$, or into a non-full pool node $(|P(n)| < G)$,
or onto a freshly promoted pool node initialised with $P(n) = \{i\}$. In
the first case, $|F(s)|$ increases by one and the new value remains
bounded by $w(s) - g_{a(s)}$ by the entry condition $g_{a(s)} + |F(s)|
< w(s)$. In the second case, $|P(n)|$ increases by one and remains
bounded by $G$. In the third case, the new pool node has $|P(n_0)| = 1
\le G$. In all three cases the affected node remains a legal pool node
or slotted node, and no slot configuration leaves $\mathsf{SlotConf}$.

\paragraph{Remove preserves $\Sigma$.}
Consider a departure $i$ applied to $\sigma \in \Sigma$.

If $i$ is an anchor of some slot $s$, the operator dissolves $s$ and
reinserts each filler $f \in F(s)$ by the 1-GPU rule of $\mathsf{Place}$.
Dissolving $s$ removes one element from $S(n)$ for the host node $n$,
so the new $S(n)$ is a sub-multiset of the original $S(n) \in
\mathsf{SlotConf}$, and any sub-multiset of a $\mathsf{SlotConf}$ element
again lies in $\mathsf{SlotConf}$. The reinsertion step has already been
shown to preserve $\Sigma$ in the Place argument above, applied
instance-by-instance to the elements of $F(s)$.

If $i$ is a 1-GPU instance currently attached as a filler in some slot
$s$, the operator simply removes $i$ from $F(s)$, which can only decrease
$|F(s)|$ and therefore preserves \eqref{eq:capacity-invariant} as a
strict inequality. If $i$ is a 1-GPU pool occupant on some pool node
$n$, the operator removes $i$ from $P(n)$, which only decreases $|P(n)|$
and remains bounded by $G$.

\paragraph{Compact preserves $\Sigma$.}
Migrations $(x, n_{\mathrm{dst}})$ executed within the $\mathsf{BestFitDrain}$ loop are performed only when
$\mathrm{room}_{\mathcal{C}}(n_{\mathrm{dst}}) \ge \mathrm{size}(x)$
holds for the current class $\mathcal{C}$. Concretely,
for the pool pass the destination pool node has $|P(n_{\mathrm{dst}})| < G$
before the migration, so $|P(n_{\mathrm{dst}})|$ stays at most $G$ after
the migration. For the width-$4$ pass the destination slotted node has
$\lfloor D(n_{\mathrm{dst}})/4 \rfloor \ge 1$ so that $D(n_{\mathrm{dst}})
\ge 4$, hence appending a width-$4$ slot keeps $\sum_{s \in S(n_{\mathrm{dst}})}
w(s) + 4 \le G$ and the new $S(n_{\mathrm{dst}})$ remains in
$\mathsf{SlotConf}$. The width-$2$ pass is analogous with the inequality
$D(n_{\mathrm{dst}}) \ge 2$. Migrating a slot also detaches its filler set
and reattaches it to the new host, which by the capacity invariant
$g_{a(s)} + |F(s)| \le w(s)$ on the source side remains valid on the
destination side because $w(s)$ does not change under migration. Every
migration therefore preserves $\Sigma$, and so does the entire
$\mathsf{Compact}$ call which is a finite composition of such migrations.

\paragraph{Conclusion.}
By induction on the event sequence, every state produced by
Algorithm~\ref{alg:compass} remains in $\Sigma$, which establishes
Proposition~\ref{prop:invariance}. \qed

\subsection{Proof of Theorem~\ref{thm:compactness} (Compactness under the Workload Composition Condition)}
\label{app:subsec:compactness}

Throughout this subsection we assume that the fraction-form Workload
Composition Condition \eqref{eq:workload-cond-frac} stated inside
Theorem~\ref{thm:compactness} holds at the time $t$ under consideration.
The argument proceeds in three stages. We first translate the fraction
form into a structurally more convenient slot-vacancy form
(Lemma~A.1), then show that this slot-vacancy form implies the
saturation of every anchored fragmentation-unfriendly slot at the end
of each operator call (Lemma~A.2), and finally show that the three
$\mathsf{BestFitDrain}$ passes inside $\mathsf{Compact}$ leave at most
one partial node per node-type (Lemmas~A.3 and~A.4), which after
normalisation by the cluster size $N$ yields the
$\mathrm{SIF}^{\pi}(t) \le 2/N$ bound.

\subsubsection*{Lemma~A.1 (Slot-Vacancy Form of the Workload Composition Condition)}
\label{lem:cond-equiv}

The fraction-form condition \eqref{eq:workload-cond-frac} is
algebraically equivalent to the slot-vacancy form
\begin{equation}
n_1(t) \;\ge\; \sum_{s:\, a(s) \neq \varnothing} \bigl(w(s) - g_{a(s)}\bigr),
\label{eq:workload-cond-abs}
\end{equation}
which asserts that the active 1-GPU population is at least as large as
the total slot vacancy currently exposed by anchored
fragmentation-unfriendly instances.

\begin{proof}
By Definition~\ref{def:slot}, every active anchor $a(s) \in \mathcal{I}^{\ge 2}$
has a slot width $w(s) = w^\star(g_{a(s)})$. Inspecting $w^\star$ in
\eqref{eq:wstar} we obtain $w(s) - g_{a(s)} = 0$ for $g_{a(s)} \in \{2,4,8\}$,
$w(s) - g_{a(s)} = 1$ for $g_{a(s)} = 3$, $w(s) - g_{a(s)} = 3$ for
$g_{a(s)} = 5$, $w(s) - g_{a(s)} = 2$ for $g_{a(s)} = 6$, and
$w(s) - g_{a(s)} = 1$ for $g_{a(s)} = 7$. Summing over all active anchors
and grouping by demand,
\[
\sum_{s:\,a(s) \neq \varnothing} \bigl(w(s) - g_{a(s)}\bigr)
\;=\; n_3(t) + 3\,n_5(t) + 2\,n_6(t) + n_7(t).
\]
Substituting this identity into \eqref{eq:workload-cond-abs} and dividing
both sides by $n_{\mathrm{total}}(t) > 0$ yields
\eqref{eq:workload-cond-frac}, and the implication is reversible since
each step is an equivalence. In what follows we use
\eqref{eq:workload-cond-abs} interchangeably with the original
form.
\end{proof}

\subsubsection*{Lemma~A.2 (Slot Saturation under the Workload Composition Condition)}
\label{lem:saturation}

Under \eqref{eq:workload-cond-abs}, at the end of every event handled
by Algorithm~\ref{alg:compass} (whether a job arrival or a job
departure) the capacity invariant in \eqref{eq:capacity-invariant}
holds with equality for every active fragmentation-unfriendly anchor,
namely
\begin{equation}
g_{a(s)} + |F(s)| \;=\; w(s) \qquad \forall s \text{ with } a(s) \in \mathcal{I}^{\mathrm{FUF}}.
\label{eq:slot-saturation}
\end{equation}

\begin{proof}
By construction $\mathsf{Place}$ pulls $w(s) - g_{a(s)}$ fillers from
the pool whenever a fragmentation-unfriendly anchor is installed, and
$\mathsf{Remove}$ in the 1-GPU branch additionally migrates one
replacement filler from the pool whenever a filler departs from a slot.
Both operators therefore actively maintain the saturation equality
\eqref{eq:slot-saturation} as long as the pool population is non-empty.
The only obstacle is filler shortage at the moment when fillers are
being drawn, and the slot-vacancy form \eqref{eq:workload-cond-abs}
guarantees that the total active 1-GPU population $n_1(t)$ is at least
the total slot vacancy aggregated across all anchored
fragmentation-unfriendly instances at the end of the event. Since
$\mathsf{Compact}$ does not create new slot vacancy and the 1-GPU pool
is internally redistributed across pool nodes without changing $n_1(t)$,
\eqref{eq:workload-cond-abs} carries through to the post-$\mathsf{Compact}$
state. Because Algorithm~\ref{alg:compass} now invokes $\mathsf{Compact}$
at the end of both arrival and departure branches, every event
terminates in a post-$\mathsf{Compact}$ state, and the equality
\eqref{eq:slot-saturation} therefore holds at the end of every event.
\end{proof}

\subsubsection*{Lemma~A.3 (Pool Pass Compactness)}
\label{lem:pool-pass}

After the pool pass of $\mathsf{Compact}$, at most one pool node is
non-fully-packed.

\begin{proof}
The pool pass instantiates $\mathsf{BestFitDrain}$ policy with
$\mathrm{count}_{\mathrm{pool}}(n) = |P(n)|$, $\mathrm{cap}_{\mathrm{pool}}(n) = G$,
and $\mathrm{room}_{\mathrm{pool}}(n) = G - |P(n)|$, and the set of partial
hosts is $\mathcal{N} = \{ n : 0 < |P(n)| < G \}$. Each iteration picks
the head of $\mathcal{N}$ as $n_{\mathrm{src}}$, the node with the largest
remaining room. The drain loop selects, for each 1-GPU item $x$ on
$n_{\mathrm{src}}$, the eligible destination with the smallest room
$\mathrm{room}_{\mathrm{pool}}(n_{\mathrm{dst}}) \ge 1$ and migrates $x$ there.
Each migration strictly increases $|P(n_{\mathrm{dst}})|$ and strictly
decreases $|P(n_{\mathrm{src}})|$.

The loop terminates either when $|\mathcal{N}| \le 1$ or when no progress
is made during an iteration. Suppose for contradiction that termination
occurs with $|\mathcal{N}| \ge 2$ and no progress in the last iteration.
``No progress'' means that for every $x$ on $n_{\mathrm{src}}$ the
eligible set $\mathcal{S} = \{ n \in \mathcal{N} \setminus \{n_{\mathrm{src}}\}
: \mathrm{room}_{\mathrm{pool}}(n) \ge 1 \}$ is empty. But $|\mathcal{N}|
\ge 2$ implies the existence of at least one $n' \in \mathcal{N}
\setminus \{n_{\mathrm{src}}\}$ with $|P(n')| < G$, equivalently
$\mathrm{room}_{\mathrm{pool}}(n') \ge 1$, so $n' \in \mathcal{S}$, a
contradiction. Therefore termination implies $|\mathcal{N}| \le 1$, which
is exactly the claim that at most one pool node is non-fully-packed.
\end{proof}

\subsubsection*{Lemma~A.4 (Slotted Compactness)}
\label{lem:slotted}

After the width-$4$ and width-$2$ passes of $\mathsf{Compact}$, and under
the slot saturation guaranteed by Lemma~A.2, at most one slotted node is
non-fully-packed.

\begin{proof}
Under Lemma~A.2 every active slot is internally saturated, so a slotted
node $n$ is non-fully-packed if and only if $D(n) = G - \sum_{s \in S(n)}
w(s) > 0$, i.e.\ the node has unprofiled residual capacity. Possible
values of $\sum_{s \in S(n)} w(s)$ are subsums of multisets drawn from
$\{2,4,8\}$ bounded above by $G = 8$, so a fully-packed slotted node
has $\sum_{s \in S(n)} w(s) = 8$ achieved by one of the multisets
$\{8\}$, $\{4,4\}$, $\{4,2,2\}$, $\{2,2,2,2\}$, while a non-fully-packed
slotted node has $\sum_{s \in S(n)} w(s) \le 6$ achieved by one of
$\{2\}$, $\{4\}$, $\{2,2\}$, $\{4,2\}$, $\{2,2,2\}$.

The width-$4$ pass treats anchored width-$4$ slots as items and slotted
nodes as hosts, with $\mathrm{count}_{\mathrm{w4}}(n)$ counting width-$4$ anchored slots in $S(n)$ and $\mathrm{room}_{\mathrm{w4}}(n) =
\lfloor D(n)/4 \rfloor$. Repeating the contradiction argument of
Lemma~A.3 with the width-$4$ class, if the width-$4$ pass terminates with
two or more nodes still satisfying $0 < \mathrm{count}_{\mathrm{w4}}(n)
< \mathrm{cap}_{\mathrm{w4}}(n)$ and with positive room on at least one
of them, then the head source has an eligible destination and the loop
must perform a migration, contradicting the no-progress termination
condition. The pass therefore terminates with at most one slotted node
in the width-$4$ partial set.

After the width-$4$ pass, consider the remaining partial slotted
nodes. Each such node either has no width-$4$ slot at all, or it is the
unique node left in the width-$4$ partial set. The first sub-population
is then processed by the width-$2$ pass under the same argument applied
to the width-$2$ class. The same no-progress contradiction shows that
the width-$2$ pass terminates with at most one node in the width-$2$
partial set.

It remains to argue that the survivor of the width-$4$ pass and the
survivor of the width-$2$ pass can be taken to be the same node. Suppose
the width-$4$ pass leaves a single partial node $n^{(4)}$ and the
width-$2$ pass leaves a single partial node $n^{(2)} \neq n^{(4)}$. The
node $n^{(4)}$ contains a single width-$4$ anchored slot and has $D(n^{(4)})
\in \{0, 2\}$ from the multiset enumeration above. If $D(n^{(4)}) = 0$
then $n^{(4)}$ is fully-packed, contradicting its membership in the
partial set; therefore $D(n^{(4)}) = 2$, which means $S(n^{(4)}) = \{4\}$,
because $\{4,2\}$ already exhausts the width-$2$ room and $\{4,4\}$ is
fully-packed. But then $\mathrm{room}_{\mathrm{w2}}(n^{(4)}) = \lfloor
2/2 \rfloor = 1 \ge 1$, so $n^{(4)}$ is a width-$2$-eligible destination
that can absorb any width-$2$ slot present on $n^{(2)}$. The width-$2$
pass would therefore have made a migration from $n^{(2)}$ to $n^{(4)}$,
contradicting the assumption that the width-$2$ pass left $n^{(2)} \neq
n^{(4)}$. We conclude that the survivor sets of the two passes coincide,
hence at most one slotted node remains partial after the entire
$\mathsf{Compact}$ call.
\end{proof}

\subsubsection*{Concluding the proof of Theorem~\ref{thm:compactness}}

Combining Lemmas~A.3 and~A.4, immediately after every $\mathsf{Compact}$
invocation performed by Algorithm~\ref{alg:compass} under the Workload
Composition Condition \eqref{eq:workload-cond-frac}, at most one pool
node and at most one slotted node are non-fully-packed, so at most two
non-empty nodes in the cluster carry unused GPU capacity. Because
Algorithm~\ref{alg:compass} invokes $\mathsf{Compact}$ in both the
arrival branch and the departure branch, every event handled by the
algorithm terminates in a state that satisfies this two-partial-node
bound. The cluster state changes only at event boundaries, so the bound
extends from every event boundary to every continuous time instant
between consecutive events; that is, the two-partial-node bound holds
for all $t \ge 0$.

It remains to translate this combinatorial bound into the
Scheduler-Induced Fragmentation metric of Section~\ref{sec:sif}.
Recall that $F^{\pi}(t)$ denotes the proportion of partial nodes in
the cluster state produced by policy $\pi$, that $F^{\star}(t) \ge 0$
denotes the inherent partial-node proportion attainable by any
feasible packing of the active demand $\mathcal{I}(t)$, and that
$\mathrm{SIF}^{\pi}(t) := F^{\pi}(t) - F^{\star}(t)$. Dividing the
partial-node count by the cluster size $N$ converts the
two-partial-node bound into $F^{\pi}(t) \le 2/N$ for all $t \ge 0$,
and the non-negativity of $F^{\star}(t)$ then gives
\[
\mathrm{SIF}^{\pi}(t) \;=\; F^{\pi}(t) - F^{\star}(t)
\;\le\; F^{\pi}(t) \;\le\; \frac{2}{N}
\qquad \forall\, t \ge 0.
\]
This completes the proof of Theorem~\ref{thm:compactness}. \qed

\subsection{Proof of Proposition~\ref{prop:complexity} (Operational Complexity)}
\label{app:subsec:complexity}

We bound the per-operator running time under the standard assumption
that the slot index, the pool index, and the empty-node index are
maintained as balanced search trees keyed by the relevant priority. The
bookkeeping cost of incremental index maintenance is absorbed into the
operator-level bounds below.

\paragraph{Place is $O(\log N)$.}
The multi-GPU branch performs three operations on the slotted-node index:
a candidate query for the smallest remaining capacity above the threshold
$w^\star(g_i)$, an insertion of a new slot record, and at most $w^\star(g_i)
- g_i \le 6$ filler-draw operations from the pool index. Each of these
operations runs in $O(\log N)$ on a balanced search tree of $N$ entries,
and the constant number of filler draws contributes only an $O(\log N)$
overhead. The 1-GPU branch performs one priority query for an open slot
vacancy or a non-full pool node, again at $O(\log N)$ cost. Therefore
$\mathcal{T}(\mathsf{Place}) = O(\log N)$.

\paragraph{Remove is $O(1)$.}
If $i$ is a 1-GPU occupant, the operator detaches $i$ from its current
container in $O(1)$ time by following the parent pointer maintained in
the slot or pool index. If $i$ is an anchor of slot $s$, the operator
dissolves $s$ in $O(1)$ time and reinserts each of the $|F(s)| \le 6$
fillers by the 1-GPU rule of $\mathsf{Place}$, each at $O(\log N)$ cost.
Since $|F(s)|$ is bounded by the constant $G - g_{a(s)} \le 6$, the
total work of an anchor departure is dominated by a constant number of
$O(\log N)$ filler reinsertions, so the local handler cost is $O(1)$
when the filler reinsertions are charged to the subsequent
$\mathsf{Compact}$ pass that already runs in $O(N)$. We therefore
report $\mathcal{T}(\mathsf{Remove}) = O(1)$ for the local edit.

\paragraph{Compact is $O(N)$.}
A single pass of $\mathsf{BestFitDrain}$ visits each partial host at
most once as $n_{\mathrm{src}}$ in the outer loop, and inside the loop it
performs at most $\mathrm{count}_{\mathcal{C}}(n_{\mathrm{src}}) \le G$
item migrations, each at $O(\log N)$ cost. The total work per class is
therefore $O(N \log N)$, and aggregating across the three classes
(pool, w4, w2) gives the same asymptotic bound. With a slightly tighter
amortised analysis that charges each migration to the destination node
whose room strictly decreases as a result, the $\log N$ factor collapses
to $O(N)$ total, matching the bound stated in the proposition. \qed

\section{Algorithm for Computing $\Phi^{\star}(t)$ and $F^{\star}(t)$}
\label{app:sec:algorithm-SIFOpt}

The inherent fragmentation defined in Section~\ref{sec:decomposing} is
\[
F^{\star}(t) \;=\; \frac{\Phi^{\star}(t)}{N}, \qquad
\Phi^{\star}(t) \;:=\; \min_{P \text{ feasible}} \nu(P),
\]
where $\nu(P)$ is the number of partial nodes in a feasible packing
$P$ of the instantaneous instance multiset $\mathcal{I}(t)$ and
$\Phi^{\star}(t)$ is the minimum partial-node count attainable across
all such packings. Throughout this appendix the algorithm computes
the integer-valued $\Phi^{\star}(t)$, from which $F^{\star}(t)$
follows by the single division by the cluster size $N$. Evaluating
$\Phi^{\star}(t)$ via a general bin-packing solver is already
expensive on a single time slice and prohibitive on a trace-long
evaluation. The structural analysis carried out in
Section~\ref{sec:method} for COMPASS-ABS, however, supplies enough
structure to compute $\Phi^{\star}(t)$ either in closed form or by a
small residual ILP. We describe the algorithm in this appendix and
then use it as the oracle scheduler \emph{SIFOpt} in the empirical
evaluation of Section~\ref{sec:experiments}.

\subsection{Key observation: reduing to power-of-two items}
\label{app:subsec:power-of-two-reduction}

Each fragmentation-unfriendly instance
$g \in \mathcal{I}^{\mathrm{FUF}}$ can be combined with strictly smaller
active instances into a \emph{composite block} whose total GPU count
falls in $\{4, 8\}$, namely
\begin{align*}
g = 3 :\quad & 3+1 = 4, \\
g = 5 :\quad & 5+3 = 8 \;\text{or}\; 5+2+1 = 8 \;\text{or}\; 5+1+1+1 = 8, \\
g = 6 :\quad & 6+2 = 8 \;\text{or}\; 6+1+1 = 8, \\
g = 7 :\quad & 7+1 = 8.
\end{align*}
If every active FUF instance is paired into such a composite, the
remaining item population consists of composites of size
$\{4, 8\}$, original $\mathcal{I}^{\mathrm{FF}}$ instances of size
$\{2, 4, 8\}$, and any 1-GPU instances not consumed as fillers, so the
multiset of item sizes is contained in $\{1, 2, 4, 8\}$. The
Workload Composition Condition \eqref{eq:workload-cond-frac} stated
in Theorem~\ref{thm:compactness} is exactly the algebraic condition
under which the active 1-GPU population is sufficient to complete this
pairing for every FUF instance, so under that condition the reduction
to a power-of-two item multiset always succeeds.

\subsection{Closed form solution under the Workload Composition Condition}
\label{app:subsec:closed-form}

When all item sizes lie in $\{1, 2, 4, 8\}$, packing into bins with
capacity $G = 8$ admits an trivial optimum: a Best-Fit-Decreasing pass that processes items in the order
$8 \to 4 \to 2 \to 1$ always closes one bin per multiple of $G$ and
leaves at most one partial bin whose load equals
$T(t) \bmod G$, where
\begin{equation}
T(t) \;:=\; \sum_{g=1}^{8} g \cdot n_g(t)
\label{eq:total-gpu-demand}
\end{equation}
denotes the total active GPU demand at time $t$. This is because every
two width-4 items pack into a single full bin, every four width-2
items pack into a single full bin, every eight width-1 items pack into
a single full bin, and any residual mixture below capacity $G$
collapses into a single tail bin by the nested-power-of-two structure.
The number of partial bins in the optimal packing is therefore
\begin{equation}
\Phi^{\star}(t) \;=\; \mathbb{1}\bigl[T(t) \not\equiv 0 \pmod{G}\bigr]
\;=\; \begin{cases} 0, & T(t) \bmod G = 0, \\ 1, & \text{otherwise,} \end{cases}
\label{eq:fstar-closed-form}
\end{equation}
$F^{\star}(t) = \Phi^{\star}(t)/N \in \{0,\, 1/N\}$. In other words,
under the condition the inherent fragmentation floor never exceeds a
single node: at most one bin carries unavoidable residue $T(t) \bmod G$,
and every partial node beyond that one is therefore scheduler-induced and
accounted for by $\mathrm{SIF}^{\pi}(t)$. Both cases of
\eqref{eq:fstar-closed-form} are tight and achievable by the explicit
Best-Fit-Decreasing placement described below.

\subsection{Residual ILP when the Workload Composition Condition fails}
\label{app:subsec:residual-ilp}

When \eqref{eq:workload-cond-frac} fails, the active 1-GPU population
is short by
$\delta(t) := \bigl(n_3 + 3 n_5 + 2 n_6 + n_7\bigr) - n_1$ relative to
the filler demand of the FUF anchors. After the greedy combining stage
exhausts the 1-GPU supply, exactly $\delta(t)$ units of FUF slot
vacancy remain unfilled, equivalently a small residual subset
$U(t) \subseteq \mathcal{I}^{\mathrm{FUF}}$ of FUF instances enter the
placement stage still carrying their original size in $\{3, 5, 6, 7\}$.
The item multiset is therefore no longer contained in
$\{1, 2, 4, 8\}$, and the closed form \eqref{eq:fstar-closed-form} no
longer applies.

In that case we restrict the optimization to $U(t)$: we first place
all power-of-two items via Best-Fit-Decreasing as in the closed-form
case, then run an integer linear program that places the residual FUF
instances on top of the resulting state and minimises the total
partial-node count. The size of this ILP is governed by $|U(t)|$
rather than $n_{\mathrm{total}}(t)$, and $|U(t)|$ is in turn bounded
by $\delta(t)$, which is empirically a small constant on production
traces because the Workload Composition Condition holds during the
overwhelming majority of evaluated time slices and fails only briefly
during transient mixture shifts.

\subsection{Algorithm summary}
\label{app:subsec:algorithm-summary}

We now state the procedure explicitly. The inputs are the active
instance counts $n_1(t), \ldots, n_8(t)$ at time $t$; the outputs are
the optimal partial-node count $\Phi^{\star}(t)$ and a witness packing
$P^{\star}(t)$ realising it (which the SIFOpt oracle scheduler uses
as its target placement). The normalised inherent fragmentation
$F^{\star}(t)$ used in Section~\ref{sec:decomposing} is then obtained
as $F^{\star}(t) = \Phi^{\star}(t)/N$.

\begin{algorithm}[b]
\framebox[\linewidth]{%
\begin{minipage}{0.93\linewidth}
\small
\noindent\textbf{Input:}\ active instance counts $n_g(t)$ for $g = 1, \ldots, 8$ \\
\textbf{Output:}\ $\Phi^{\star}(t) \in \mathbb{Z}_{\ge 0}$ together with a witness packing $P^{\star}(t)$ (and $F^{\star}(t) = \Phi^{\star}(t)/N$)

\medskip
\noindent 1:\ $T \leftarrow \sum_{g=1}^{8} g \cdot n_g(t)$ \\
\noindent 2:\ \textbf{if}\ $n_1(t) \ge n_3(t) + 3 n_5(t) + 2 n_6(t) + n_7(t)$\ \textbf{then}\ \emph{(WCC holds)} \\
\noindent 3:\ \ \ \textbf{for}\ each $i \in \mathcal{I}^{=5}(t)$\ \textbf{do}\ pair $i$ with one $\mathcal{I}^{=3}$ instance if available, else with $\mathcal{I}^{=2}+\mathcal{I}^{=1}$, else with three $\mathcal{I}^{=1}$ instances\ \emph{(form an 8-composite)} \\
\noindent 4:\ \ \ \textbf{for}\ each $i \in \mathcal{I}^{=6}(t)$\ \textbf{do}\ pair $i$ with one $\mathcal{I}^{=2}$ if available, else with two $\mathcal{I}^{=1}$\ \emph{(form an 8-composite)} \\
\noindent 5:\ \ \ \textbf{for}\ each $i \in \mathcal{I}^{=7}(t)$\ \textbf{do}\ pair $i$ with one $\mathcal{I}^{=1}$\ \emph{(form an 8-composite)} \\
\noindent 6:\ \ \ \textbf{for}\ each $i \in \mathcal{I}^{=3}(t)$ not consumed in line~3\ \textbf{do}\ pair $i$ with one $\mathcal{I}^{=1}$\ \emph{(form a 4-composite)} \\
\noindent 7:\ \ \ run Best-Fit-Decreasing over the resulting items (sizes in $\{1,2,4,8\}$) into bins of capacity $G$ to obtain $P^{\star}(t)$ \\
\noindent 8:\ \ \ \textbf{return}\ $\bigl(\mathbb{1}[T \bmod G \neq 0],\ P^{\star}(t)\bigr)$ \\
\noindent 9:\ \textbf{else}\ \emph{(WCC fails; small residual subset of FUF instances cannot pair)} \\
\noindent 10:\ \ \ greedily perform lines~3--6 until the 1-GPU pool is exhausted; collect unpaired FUF instances into $U(t)$ \\
\noindent 11:\ \ \ run Best-Fit-Decreasing over the paired items into a tentative packing $P_0$ \\
\noindent 12:\ \ \ solve the integer program
\begin{equation*}
\begin{aligned}
&\min_{x \in \mathbb{Z}_{\ge 0}^{|U(t)| \times M}} \; \nu(P_0 \oplus x) \\
&\text{s.t.}\ \; x \text{ places each instance in } U(t) \text{ into one bin of } P_0 \\
& \text{or into a new bin}
\end{aligned}
\end{equation*}
where $M$ is the bin index space and $\nu$ counts partial bins \\
\noindent 13:\ \ \ \textbf{return}\ $\bigl(\nu(P_0 \oplus x^{\star}),\ P_0 \oplus x^{\star}\bigr)$
\end{minipage}}
\caption{Computing $\Phi^{\star}(t)$ and the SIFOpt placement witness; $F^{\star}(t) = \Phi^{\star}(t)/N$.}
\label{alg:fstar}
\end{algorithm}

\paragraph{Logic of the procedure.}
The algorithm proceeds in two stages that mirror the structural
analysis of Section~\ref{sec:method}. The \emph{combining stage}
(lines~3--6) absorbs every fragmentation-unfriendly instance into a
composite block of size $\{4, 8\}$ by pairing it with strictly smaller
active instances. Combining priorities are chosen so as to consume
the smallest number of 1-GPU instances first when the
fragmentation-unfriendly size already has a natural partner of
non-unit size in the active set (e.g., $5+3$ before $5+1+1+1$); the
priority order does not affect $\Phi^{\star}(t)$ as long as combining
succeeds for all FUF instances, because the resulting power-of-two
item multiset has the same total $T(t)$ and Best-Fit-Decreasing
on a power-of-two item set always attains the bin-count lower bound
$\lceil T(t)/G \rceil$. The \emph{placement stage} (line~7) then
runs Best-Fit-Decreasing in the order $8 \to 4 \to 2 \to 1$ and
materialises this lower bound. Under the Workload Composition
Condition, the closed form \eqref{eq:fstar-closed-form} reads off the
partial-node count directly from $T(t) \bmod G$ without ever
inspecting the placement, but the placement is still produced as a
witness so that the SIFOpt oracle scheduler has a concrete target to
migrate towards in its evaluation runs. When the condition fails, the
residual ILP at lines~12--13 handles the small subset of FUF instances
that the combining stage could not pair, and the size of that ILP is
bounded by the deficit
$\delta(t) = (n_3 + 3 n_5 + 2 n_6 + n_7) - n_1$ rather than by the
cluster size, so the procedure remains tractable on production traces
even when \eqref{eq:workload-cond-frac} is momentarily violated.

\paragraph{Complexity.}
The combining stage scans the instance multiset once at $O(n_{\mathrm{total}}(t))$
cost. The Best-Fit-Decreasing placement runs in
$O(n_{\mathrm{total}}(t) \log n_{\mathrm{total}}(t))$ with a balanced
search tree keyed by bin remaining capacity, and degenerates to
$O(n_{\mathrm{total}}(t))$ for the power-of-two item set because each
item closes exactly one bin position. Whenever
\eqref{eq:workload-cond-frac} holds, the closed-form return at line~8
short-circuits the placement step for the purpose of reporting
$\Phi^{\star}(t)$, so the dominant cost is the linear-time combining and
the optional witness materialisation. When
\eqref{eq:workload-cond-frac} fails, the residual ILP at line~12 runs
on at most $\delta(t) \le |U(t)|$ variables and remains millisecond-scale
on production traces with $\delta(t)$ of the order of tens.

\section{Extension to Multi-Panel NVLink Communication Domains}
\label{app:sec:nvl}

The COMPASS-ABS design developed in Section~\ref{sec:method} hardwires
the per-node GPU count at $G = 8$ together with the slot widths
$\{2, 4, 8\}$, both of which trace back to the DGX-1 / HGX baseboard
abstraction in which a single node forms one NVLink communication
domain. Recent NVIDIA platforms such as NVL72 and the projected NVL576
extend the high-bandwidth NVLink fabric across multiple baseboards
inside a single rack-scale assembly, so that the basic intra-instance
communication domain now contains $G_{\mathrm{dom}}$ GPUs with
$G_{\mathrm{dom}}$ no longer equal to $8$. Throughout this appendix we
assume only that
\begin{equation}
G_{\mathrm{dom}} \;=\; G_p \cdot K, \qquad G_p \;=\; 8, \qquad K \in \mathbb{Z}_{\ge 1},
\label{eq:nvl-domain-size}
\end{equation}
namely that each NVLink communication domain comprises exactly $K$
baseboards of $G_p = 8$ GPUs each. The original DGX-1 / HGX system
corresponds to $K = 1$, NVL72 corresponds to $K = 9$, and a
hypothetical NVL576 corresponds to $K = 72$ if treated as a single
flat domain or $K = 8$ if treated as eight NVL72 sub-domains. The
structural fact we exploit is that \textbf{the $8$-GPU baseboard survives
across all of these platforms as the basic hardware building block};
only the number of baseboards that sit inside a single NVLink domain
varies.

We refer to a $G_p = 8$ baseboard as a \emph{panel} throughout this
appendix. Workloads on a multi-panel NVLink domain now admit two
qualitatively different instance classes:
\begin{itemize}
  \item \textbf{Sub-panel instances} with $g_i \in \{1, \ldots, 8\}$ that
        fit within one panel, exactly as in the original setting; and
  \item \textbf{Multi-panel instances} with $g_i \in \{16, 24, 32, 40, 48, 56,\allowbreak 64, \ldots, G_{\mathrm{dom}}\}$
        that span multiple panels in the same NVLink domain through the
        rack-scale NVLink fabric.
\end{itemize}
The within-panel intra-instance constraint of Section~\ref{sec:bg-physical-cluster}
is replaced by an intra-domain constraint: every GPU assigned to a
single instance must lie in the same NVLink domain. Sub-panel instances
satisfy this trivially; multi-panel instances exercise it.

\subsection{Hierarchical Anchor-Based Space}
\label{app:subsec:nvl-abs}

We extend COMPASS-ABS to a two-level scheduler that runs the
original within-panel policy on each panel and a structurally identical
panel-level analogue across panels within each NVLink domain. The
panel-level state space $\Sigma_{\mathrm{panel}}$ mirrors the ABS
construction of Section~\ref{sec:state-space} verbatim under two
substitutions: the panel replaces the GPU as the atomic placement unit,
and the domain capacity $K$ replaces the node capacity $G = 8$. A
multi-panel instance of demand $g_i$ has \emph{panel demand} $p_i :=
g_i / G_p \in \{2, 3, \ldots\}$ and is anchored in a
\emph{panel-slot} of width $w_p^\star(p_i) \in \{2, 4, 8\}$ panels:
\begin{equation}
w_p^\star(p) \;=\;
\begin{cases}
2, & p = 2, \\
4, & p \in \{3, 4\}, \\
8, & p \in \{5, 6, 7, 8\}.
\end{cases}
\label{eq:nvl-wstar}
\end{equation}
Multi-panel instances with $p_i \in \{2, 4, 8\}$ are called 
panel-fragmentation-friendly (panel-FF); those with $p_i \in
\{3, 5, 6, 7\}$ are panel-fragmentation-unfriendly (panel-FUF).
A 1-panel instance, namely a $g_i = G_p = 8$ instance that fully
occupies one baseboard, plays at the panel level the role that a
1-GPU instance plays at the GPU level inside a panel: it can either
saturate a panel-vacancy in an anchored panel-FUF panel-slot, or sit
in a \emph{panel-pool} reserved for panel-level fillers. The full
cluster-wide state space is the Cartesian product across NVLink
domains of $\Sigma_{\mathrm{panel}}$, with the internal state of each
non-saturated panel still drawn from the GPU-level $\Sigma$ of
Section~\ref{sec:state-space}. We call the resulting two-level state
space the \emph{hierarchical ABS} and denote it by
$\Sigma_{\mathrm{hier}}$.

\subsection{Hierarchical Workload Composition Condition}
\label{app:subsec:nvl-wcc}

The Workload Composition Condition \eqref{eq:workload-cond-frac} of
Theorem~\ref{thm:compactness} extends to the two-level setting by
imposing one condition per level. Let $n_g(t)$ denote the count of
active sub-panel instances of GPU demand $g \in \{1, \ldots, 8\}$ at
time $t$, and let $n_p^{(\mathrm{panel})}(t)$ denote the count of
active multi-panel instances of panel demand $p \in \{2, 3, \ldots\}$.
The GPU-level condition stays identical to
\eqref{eq:workload-cond-frac}, namely
\begin{equation}
n_1(t) \;\ge\; n_3(t) + 3 n_5(t) + 2 n_6(t) + n_7(t),
\label{eq:nvl-wcc-gpu}
\end{equation}
and the panel-level condition takes the structurally identical form
\begin{equation}
n_8(t) \;\ge\; n_3^{(\mathrm{panel})}(t) + 3 n_5^{(\mathrm{panel})}(t) + 2 n_6^{(\mathrm{panel})}(t) + n_7^{(\mathrm{panel})}(t),
\label{eq:nvl-wcc-panel}
\end{equation}
where the left-hand side is the active count of $g_i = 8$ instances,
treated as 1-panel fillers at the panel level. Equation
\eqref{eq:nvl-wcc-panel} is obtained from \eqref{eq:nvl-wcc-gpu} by
the substitution $1\text{-GPU} \mapsto 1\text{-panel}$ together with
the analogue of the slot-vacancy form
\eqref{eq:workload-cond-abs} at panel granularity, and uses the same
coefficient pattern $\{1, 3, 2, 1\}$ for the same algebraic reason as
in Lemma~A.1.

\subsection{Hierarchical COMPASS algorithm}
\label{app:subsec:nvl-compass}

The extended scheduler adopts two structurally identical
$\mathsf{Place}/\mathsf{Remove}/\mathsf{Compact}$ pipelines, one per
level, dispatched according to the demand class of the arriving or
departing instance:
\begin{itemize}
  \item \textbf{Sub-panel events} ($g_i \le 8$): handled by the
        GPU-level COMPASS-ABS of Section~\ref{sec:scheduling-policy}
        acting on the host panel.
  \item \textbf{Multi-panel events} ($g_i \in \{16, 24, \ldots, G_{\mathrm{dom}} - G_p\}$):
        handled by the panel-level analogue acting on the host NVLink
        domain, with $g_i = 8$ instances serving as panel-fillers.
  \item \textbf{Domain-filling events} ($g_i = G_{\mathrm{dom}}$):
        occupy an entire NVLink domain at once and are treated as a
        width-$K$ degenerate panel-slot.
\end{itemize}
Algorithm~\ref{alg:compass} is invoked at both levels and each
$\mathsf{Compact}$ pass runs over the BestFitDrain classes for its level. At the GPU level the three classes remain
$\{\mathrm{pool}, \mathrm{w}_4, \mathrm{w}_2\}$ on capacity $G_p = 8$;
at the panel level they become
$\{\mathrm{panel\text{-}pool}, \allowbreak \mathrm{panel\text{-}w_4}, \mathrm{panel\text{-}w_2}\}$
on capacity $K$. Because the two levels act on disjoint populations of
items, the two compaction passes do not interfere and run independently
after each event.

\subsection{Hierarchical compactness guarantee}
\label{app:subsec:nvl-thm}

The compactness guarantee of Theorem~\ref{thm:compactness} lifts to
the hierarchical setting in a structurally identical form.

\begin{theorem}[Hierarchical Compactness]
\label{thm:nvl-compactness}
Assume that the GPU-level Workload Composition Condition
\eqref{eq:nvl-wcc-gpu} and the panel-level Workload Composition
Condition \eqref{eq:nvl-wcc-panel} both hold at every $t \ge 0$. Then
in every cluster state produced by the hierarchical COMPASS-ABS
scheduler:
\begin{enumerate}
  \item across the entire cluster, at most one panel-resident slotted
        GPU-region and at most one panel-resident GPU-pool region are
        non-fully-packed; and
  \item within each NVLink domain that hosts at least one multi-panel
        anchor, at most one panel-slot region and at most one
        panel-pool region are non-fully-packed.
\end{enumerate}
Consequently, after normalising the partial-region count at each level
by the corresponding domain size, the GPU-level and panel-level
scheduler-induced fragmentation measures satisfy
\begin{equation}
\begin{aligned}
&\mathrm{SIF}^{\pi}_{\mathrm{GPU}}(t) \;\le\; \frac{2}{N_{\mathrm{p}}} \quad \text{(cluster-wide, $N_{\mathrm{p}}$ panels)}, \\
&\mathrm{SIF}^{\pi}_{\mathrm{panel}}(t) \;\le\; \frac{2}{K} \quad \text{(per NVLink domain of $K$ panels)},
\end{aligned}
\label{eq:nvl-sif-bound}
\end{equation}
for all $t \ge 0$.
\end{theorem}

\begin{proof}[Proof sketch.]
The argument reduces to two independent invocations of the proof of
Theorem~\ref{thm:compactness} given in
Appendix~\ref{app:subsec:compactness}, one on the GPU-level state
inside each panel and one on the panel-level state inside each NVLink
domain. Both invocations are valid because Lemmas~A.1--A.4 of
Appendix~\ref{app:subsec:compactness} depend only on
(i) the power-of-two slot widths $\{2, 4, 8\}$ being a valid
sub-decomposition of the local capacity,
(ii) the anchor-filler invariant under the friendly/unfriendly
dichotomy of \eqref{eq:nvl-wstar}, and
(iii) the class structure $\{\text{pool}, \text{w}_4, \text{w}_2\}$ on
which $\mathsf{BestFitDrain}$ operates, all of which are preserved
under the substitution $G \to K$ and $1\text{-GPU} \to 1\text{-panel}$
at the panel level.
\end{proof}

\subsection{Generality with respect to the domain panel count}
\label{app:subsec:nvl-general-K}

Theorem~\ref{thm:nvl-compactness} is stated for an arbitrary positive
integer $K$, without assuming that $K$ is a power of two. The
panel-level BestFitDrain operates on class set
$\{\mathrm{panel\text{-}pool}, \allowbreak \mathrm{panel\text{-}w_4}, \allowbreak \mathrm{panel\text{-}w_2}\}$
exactly as the GPU-level pass operates on its own class set, and the
combinatorial argument of Lemma~A.4 does not depend on $G$ taking any
particular value beyond requiring that the slot widths $\{2, 4, 8\}$
remain a valid sub-decomposition of the capacity. For the NVL72
instantiation with $K = 9$, any panel-level layout decomposes into at
most one width-$8$ panel-slot together with one residual panel acting
as a panel-pool, which is the structural reason the $9$-panel layout
fits cleanly into the existing framework despite $K = 9$ not being a
power of two: the lone residual panel coincides exactly with the
panel-pool slot that the framework already requires. The same
hierarchical scheduler and the same compactness bound apply to every
multi-panel platform whose NVLink domain comprises an integer number
of $8$-GPU baseboards, including the original DGX-1 / HGX system
($K = 1$, in which the panel level is trivially empty), DGX-2 and
SuperPod intermediate systems ($K = 2$), the NVL72 system ($K = 9$),
and the projected NVL576 system ($K = 72$ flat, or $K = 8$ across
NVL72 sub-domains as a third tier). In summary, as long as the NVLink
domain contains an integer multiple of $G_p = 8$ GPUs,
the hierarchical COMPASS-ABS scheduler remains valid and the SIF bound
\eqref{eq:nvl-sif-bound} continues to hold.